\documentclass[12pt]{article}
\usepackage{enumerate}
\usepackage{natbib}
\newcommand{\blind}{1}

\usepackage{url}
\usepackage{amsmath,amsxtra,amssymb,latexsym,epsfig,amscd,amsthm,fancybox,epsfig}
\usepackage{amsfonts}
\usepackage{mathtools}
\usepackage{float}
\restylefloat{figure}
\usepackage{booktabs}
\usepackage{multirow}
\usepackage{color}
\usepackage{placeins}
\usepackage{lscape}
\usepackage{mathtools}
\usepackage{amsmath}
\usepackage{amssymb}
\usepackage{comment}

\usepackage[mathscr]{eucal}
\usepackage{graphicx}
\usepackage{tikz}
\usetikzlibrary{calc}
\usepackage{caption}
\usepackage{multicol,xcolor}
\usepackage{epsfig} % for postscript graphics files
\usepackage{epstopdf}
\usepackage{cases}
\usepackage{color}
\usepackage{hyperref}
\usepackage{bm}  %JLK
\usepackage{bbm} %JLK
\usepackage[ruled,vlined]{algorithm2e}
\usepackage{algorithmic}

\usepackage{color, colortbl}
\usepackage{booktabs,xcolor}

\newtheorem{thm}{Theorem}[section]

\newtheorem{remark}{Remark}[section]
\newtheorem{cor}{Corollary}[section]

\newtheorem{rmk}{Remark}[section]

\theoremstyle{definition}

\theoremstyle{remark}

\numberwithin{equation}{section}

\newcommand{\M}{\mathcal{M}}

\newcommand{\E}{\mathbb{E}}

\newcommand{\norm}[1]{\left\Vert#1\right\Vert}

\numberwithin{equation}{section}

\newcommand{\bed}{\begin{displaymath}}
\newcommand{\eed}{\end{displaymath}}
\newcommand{\bea}{\bed\begin{array}{rl}}
\newcommand{\eea}{\end{array}\eed}

\newcommand{\barray}{\begin{array}{ll}}
\newcommand{\earray}{\end{array}}
\newcommand{\diag}{{\rm diag}}

\newcounter{mntcomm}

\def\bar{\overline}
\def\a.s{\text{\;a.s.\;}}

\newcommand{\beq}[1]{\begin{equation} \label{#1}}
\newcommand{\eeq}{\end{equation}}

\def\KL{\text{\rm KL}}
\def\LB{\text{\rm LB}}
\def\Ber{\text{\rm Bernoulli}}
\def\Beta{\text{\rm Beta}}
\def\vec{\text{\rm vec}}
\def\vech{\text{\rm vech}}
\def\Pois{\text{\rm Poisson}}

\begin{document}

\def\spacingset#1{\renewcommand{\baselinestretch}%
{#1}\small\normalsize} \spacingset{1}

%%%%%%%%%%%%%%%%%%%%%%%%%%%%%%%%%%%%%%%%%%%%%%%%%%%%%%%%%%%%%%%%%%%%%%%%%%%%%%

\if1\blind
{
  \title{\bf Natural Gradient Variational Bayes without Fisher Matrix Analytic Calculation and Its Inversion}
  \author{A. Godichon-Baggioni\\
  %\thanks{
    %The authors gratefully acknowledge \textit{please remember to list all relevant funding sources in the unblinded version}}\hspace{.2cm}\\
    Laboratoire de Probabilit\'{e}s 
 Statistique et Mod\'{e}lisation\\
 Sorbonne-Université, 75005, Paris, France\\
Email: antoine.godichon\_baggioni@sorbonne-universite.fr\\
  %  and \\
    D. Nguyen \\
    Department of Mathematics\\
Marist College, 3399 North Road\\
Poughkeepsie NY 12601, United States\\
Email: nducduy@gmail.com\\
%and\\
M-N. Tran\\
Business Analytics discipline\\
The University of Sydney Business School\\
Australia\\
Email: minh-ngoc.tran@sydney.edu.au
   }
  \maketitle
} \fi

\if0\blind
{
  \bigskip
} \fi

\bigskip
\begin{abstract}
This paper introduces a method for efficiently approximating the inverse of the Fisher information matrix, a crucial step in achieving effective variational Bayes inference.
%This paper presents an approach for efficiently approximating the inverse of Fisher information, a key component in deriving efficient variational Bayes inference. 
A notable aspect of our approach is the avoidance of analytically computing the Fisher information matrix and its explicit inversion. Instead, we introduce an iterative procedure for generating a sequence of matrices that converge to the inverse of Fisher information. 
The natural gradient variational Bayes algorithm without 
analytic expression of the Fisher matrix and its inversion is provably convergent and achieves a convergence rate of order $\mathcal O(\log s/s)$, with $s$ the number of iterations.
We also obtain a central limit theorem for the iterates.
Implementation of our method does not require storage of large matrices, and achieves a linear complexity in the number of variational parameters.
Our algorithm exhibits versatility, making it applicable across a diverse array of variational Bayes domains, including Gaussian approximation and normalizing flow Variational Bayes. 
We offer a range of numerical examples to demonstrate the efficiency and reliability of the proposed variational Bayes method.
\end{abstract}

\noindent%
\vfill

%\newpage
%\spacingset{1.9} % DON'T change the spacing!
\section{Introduction}

The growing complexity of models used in modern statistics and machine learning has spurred the demand for more efficient Bayesian estimation techniques. Among the array of Bayesian tools available, Variational Bayes \citep{waterhouse1995bayesian,jordan1999introduction} has gained prominence as a remarkably versatile alternative to traditional Monte Carlo methods for tackling statistical inference in intricate models. Variational Bayes (VB) operates by approximating the posterior probability distribution using a member selected from a family of tractable distributions, characterized by variational parameters. The optimal member is determined through minimization of the Kullback-Leibler divergence, which quantifies the disparity between the chosen candidate and the posterior distribution. 
VB is a fast alternative to Markov chain Monte Carlo (MCMC) methods,
and has found diverse applications, encompassing variational autoencoders \citep{kingma2013auto}, text analysis \citep{hoffman2013stochastic}, Bayesian synthetic likelihood \citep{ong2018variational}, deep neural networks \citep{graves2011practical,Tran:JCGS2020}, to name a few. For recent advances in the field of VB and Bayesian approximation in general, please refer to 
the excellent survey papers of \cite{blei2017variational} and \cite{martin2023approximating}.

VB turns the Baysesian inference problem into an optimization problem,
and a large class of VB methods use stochastic gradient descent (SGD) as their backbone.
%The introduction of stochastic gradient methods, first pioneered by \cite{robbins1951stochastic}, has significantly bolstered the applicability of variational Bayes, particularly in handling high-dimensional models and extensive datasets, as evidenced by the work of \cite{hoffman2013stochastic} and \cite{salimans2013fixed}.  
In the recent decades, a great deal of effort has been devoted to 
developing and improving optimization algorithms for big and high dimensional data.
As a result, various first order stochastic optimization
algorithms have been developed in response to these new demands; notable examples include 
AdaGrad of \cite{duchi2011adaptive},
Adam of \cite{kingma2014adam}, Adadelta of \cite{zeiler2012adadelta},
and their variance reduction variations \citep{defazio2014saga,johnson2013accelerating,nguyen2017sarah}.
For a detailed discussion on stochastic optimization, please refer to the excellent books of \cite{kushner2003stochastic,goodfellow2016deep} and \cite{murphy2012machine}.

Gradient descent methods in VB rely on the gradient of the objective lower bound function,
whose definition depends upon the metric on the variational parameter space. 
Optimization in conventional VB methods uses the Euclidean gradient defined using the usual Euclidean metric.
It turns out that the natural gradient, the term coined by \cite{amari1998natural}, represents a more adequate direction of ascent in the VB context as it takes into account the information geometry of the variational family \citep{martens2020new,khan2017conjugate}. 
The natural gradient is defined using the Fisher-Rao metric, 
which resembles the Kullback-Leibler divergence between probability distributions parameterized by the variational parameters.
More precisely, the natural gradient is the steepest ascent direction of the objective function on the variational parameter space equipped with the Fisher-Rao metric.
\cite{martens2020new} sheds light on the concept that natural gradient descent can be viewed as a second-order optimization method where the Fisher information assumes the role of the Hessian matrix.
Because of this, natural gradients take into account the curvature information (through the Fisher-Rao metric) of the variational parameter space; therefore the number of iteration steps required to find a local optimum is often found significantly reduced \citep{tran2017variational}. According to \cite{tan2021analytic}, stochastic optimization guided by natural gradients has proven more resilient, capable of circumventing or escaping plateaus, ultimately resulting in faster convergence; see also\citep{rattray1998natural,hoffman2013stochastic,khan2017conjugate,wilkinson2023bayes}.

 The natural gradient is calculated by pre-multiplying the Euclidean gradient of the lower bound function with the inverse Fisher information matrix, a process that is notably intricate.
 Computing the Fisher matrix, not to mention its inverse, is challenging.
 In the realm of Gaussian variational approximation, where the posterior is approximated by a Gaussian distribution,
 the natural gradient can be calculated efficiently.
\cite{Tran:JCGS2020} consider a factor structure for the covariance matrix, and derive a closed-form approximation for the natural gradient.
\cite{tan2021analytic} employs a Cholesky factor structure for the covariance matrix and the precision matrix, and derives an analytic natural gradient; see also \cite{khan2017conjugate} and \cite{magris2022exact}.
On the other hand, for broader cases where the variational distribution is based on neural networks, \cite{martens2015optimizing} approximate the Fisher matrix with a block diagonal matrix. 
It is important to highlight that existing techniques for computing the natural gradient are primarily restricted to certain contexts (like the Gaussian variational approximations mentioned above) or are heavily dependent on simplified approximations (such as employing a block-diagonal matrix). These constraints restrict the broader application of natural gradients.
For a large class of VB methods, e.g., when the variational distribution is a mixture \citep{giordani2013flexible}, a copula \citep{gunawan2023flexible} or a normalizing flow \citep{rezende2015variational}, it is challenging to use the natural gradient as the Fisher matrix is not available.

This paper makes several important contributions that significantly improve the natural gradient VB method. 
First, we present an approach for efficiently approximating the {\it inverse} of Fisher information matrix. 
We emphasize that there are two main difficulties in calculating the natural gradient: (i) analytical calculation of the Fisher matrix that often involves intractable expectations, and (ii) computing its matrix inversion.
A notable aspect of our approach is the avoidance of these two difficulties altogether.
Instead, we introduce an iterative procedure for generating a sequence of positive definite matrices that converge to the inverse of Fisher information.
Pre-multiplying the Euclidean gradient with these matrices provides estimates of the natural gradient.
Our method of approximating the natural gradient is general, easy to implement, asymptotically exact and applies to any variational distribution including Gaussian distributions, mixtures and normalizing flow based distributions. 
It is important to note that, for high-dimensional applications, implementation of our method does not require storage of large matrices because the estimate of inverse Fisher matrix can be written using  outer products. Second, we propose a VB method that streamlines the natural gradient estimation without matrix inversion within the VB training iteration.
This leads to an efficient natural gradient VB algorithm, referred to as inversion-free variational Bayes (IFVB).
We also present a weighted averaged estimate version of IFVB, called AIFVB, that converges faster than IFVB.
Both IFVB and AIFVB are provably convergent, with AIFVB being shown asymptotically efficient and achieving a central limit theorem.
Third, to substantiate the effectiveness and robustness of our proposed method, we offer a range of numerical examples to  demonstrate its efficiency and reliability.

The rest of the paper is organized as follows: 
Section \ref{Variational Bayes} provides a brief overview of the variational Bayesian inference problem.
Section \ref{Natural Gradient} presents 
natural gradient and discusses its advantages as well as its computational difficulty. 
We introduce inversion free natural variational Bayes in Section \ref{Inversion Free Natural Variational Bayes}.
Section \ref{Convergence analysis} is concerned with convergence analysis.
Numerical examples are provided 
in Section \ref{Numerical Example}.  Section \ref{Conclusion} concludes the paper.
Section \ref{sec:example 5} provides more numerical examples, and 
Section \ref{sec:Main proofs} contains the proofs of the main theorems;
further technical details are in the Appendix.

\vspace{0.5cm}
\noindent{\bf Notation.} We denote by $\norm{x}=(x_1^2+\ldots+x_d^2)^{1/2}$ the $\ell^2$-norm of the vector $x=(x_1,\ldots,x_d)^\top\in\mathbb R^d$.
For a function $f$ on $\mathbb R^d$, $\nabla_x f=(\frac{\partial f}{\partial x_1},\ldots,\frac{\partial f}{\partial x_d})^\top$ denotes the gradient vector,
and $\nabla^2_x f=(\frac{\partial^2 f}{\partial x_i \partial x_j})_{i,j=1,\ldots,d}$ is the Hessian. 
$\norm{A}_{op}=\max_{\norm{x}\leq 1}\norm{Ax}$ denotes the operator norm of a matrix $A$; $\lambda_{min}(A),\lambda_{max}(A)$ denote the minimum eigenvalue and maximum
eigenvalue of matrix $A$, respectively. 
$\mathbb I_d$ denotes a $d\times d$ identity matrix. $\mathbb E_f(g(X))=\int g(x) f(x)dx$ with $X\sim f$.
We write $a=\mathcal{O}(b)$ to denote $a\leq C b$ for some constant $C>0$, and $f(x)=o(g(x)$ means $|f(x)|\leq \epsilon|g(x)|$ for all $\epsilon>0$. We use $\mathcal N(\mu,\Sigma)$
to denote a Gaussian random variable, or a Gaussian distribution, with mean $\mu$ and covariance $\Sigma$.
%=================================================%
\section{Variational Bayes}\label{Variational Bayes}
%=================================================%
This section gives a brief overview of the VB method.
Let $y$ be the data and $p(y|\theta)$ the likelihood function, with $\theta$ the set of model parameters.
Let $p(\theta)$ be the prior. Bayesian inference requires computing expectations with respect to the posterior distribution with density %(with respect to some reference measure such as the Lebesgue measure)
\[p(\theta|y)=\frac{p(\theta)p(y|\theta)}{p(y)},\]
where $p(y)=\int p(\theta)p(y|\theta)d\theta$ is often called the marginal likelihood.
It is often difficult to compute such expectations, partly because the density $p(\theta|y)$ itself is intractable as the normalizing constant $p(y)$ is unknown.
For simple models, Bayesian inference can be performed using Markov Chain Monte Carlo (MCMC), which estimates 
expectations with respect to $p(\theta|y)$ by sampling from it.
For models where $\theta$ is high dimensional or has a complicated structure, MCMC methods in their current development are either not applicable or very time consuming.
In the latter case, VB is an attractive alternative to MCMC. 
VB approximates the posterior $p(\theta|y)$ by a probability distribution with density $q_\lambda(\theta)$, $\lambda\in\M$ - the variational parameter space, belonging to some tractable family of distributions such as Gaussian. The best $\lambda$ is found by minimizing the Kullback-Leibler (KL) divergence of $p(\theta|y)$ from $q_\lambda(\theta)$
\begin{equation}\label{Prob:LambdaOptimization}
\lambda^*=\arg\min_{\lambda\in\M}\left\{ \KL(q_\lambda\|p(\cdot|y))=\int q_\lambda(\theta)\log\frac{q_\lambda(\theta)}{p(\theta|y)}d\theta\right\}.
\end{equation}
One can easily check that
\[\KL(q_\lambda\|p(\cdot|y)) = -\int q_\lambda(\theta)\log\frac{p(\theta)p(y|\theta)}{q_\lambda(\theta)}d\theta+\log p(y).\]
Thus minimizing  KL is equivalent to maximizing the lower bound which is  also called ELBO on $\log p(y)$
\begin{equation}
\LB(\lambda)=\int q_\lambda(\theta)\log\frac{p(\theta)p(y|\theta)}{q_\lambda(\theta)}d\theta
=\mathbb E_{ q_\lambda}\left[\log\frac{p(\theta)p(y|\theta)}{q_\lambda(\theta)} \right]
=\mathbb  E_{ q_\lambda}\left[ h_{\lambda}(\theta) \right],
\end{equation}
where $h_{\lambda}(\theta):=\log p(\theta)+\log p(y|\theta)-\log q_\lambda(\theta)$.
Using the fact that $\mathbb E_{q_\lambda}[\nabla_\lambda \log q_\lambda(\theta)]=0$,
it can be seen that
\begin{equation}\label{eq: score-function grad}
\nabla_{\lambda}\LB(\lambda )=\mathbb E_{q_\lambda}\left[\nabla_{\lambda}\log q_{\lambda}(\theta)\times h_{\lambda}(\theta) \right].
\end{equation}
One then can obtain an unbiased estimate of $\nabla_{\lambda} \LB(\lambda)$ by sampling from $q_\lambda$,
\begin{equation}\label{eq: grad est}
\widehat{\nabla_{\lambda} \LB}(\lambda)
=\frac{1}{B}\sum_{s=1}^B \nabla_{\lambda}\log q_{\lambda}(\theta_s)\times h_{\lambda}(\theta_s),
\quad \theta_s\sim q_{\lambda}(\theta),\;\;s=1,...,B.
\end{equation}
Alternative to \eqref{eq: score-function grad}, the Euclidean gradient 
$\nabla_{\lambda}\LB(\lambda )$ can be computed using the so-called reparameterization-trick method \citep{kingma2013auto,titsias2014doubly}.
The method for estimating the natural gradient proposed in this paper is applicable in both cases.

Stochastic gradient ascent (SGA) techniques are often employed to solve the maximization problem in \eqref{Prob:LambdaOptimization}.
More specifically, one can iteratively update
$\lambda$ as follows
\begin{equation}\label{SGD:lambdaUpdate}
\lambda^{(k+1)}=\lambda^{(k)}+\tau_{k+1}\widehat{\nabla_{\lambda} \LB}(\lambda^{(k)}),
\end{equation}
with the stepsize $\tau_k$ satisfying
$\sum_{k=1}^\infty \tau_k=\infty$ and $\sum_{k=1}^\infty \tau_k^2<\infty$.
The convergence of the update in \eqref{SGD:lambdaUpdate} 
has been studied in the literature  \citep[see, e.g.][]{robbins1951stochastic,spall2005introduction}.

SGA approximates the exact gradient at each iteration by an estimate using a mini-batch of the full sample (in big data settings) or
by sampling from $q_\lambda$ (as in \eqref{eq: grad est}).
This reduces computational cost, and facilitates on-the-fly (online) learning as new samples arrive. 
Note that, in practice, the data-dependent term $h_\lambda(\theta)$ is often estimated by using a mini-batch of the data.
It is documented extensively in the literature
\citep[see, e.g.][]{bercu2020efficient,Kirkyetal2021,chau2024inversion} that plain SGA as in \eqref{SGD:lambdaUpdate} can lead to unsatisfactory 
estimates, as it is highly sensitive to the choice
of hyper-parameters such as the step size or mini-batch size.
 In addition,
SGA is known to have slow convergence
when the Hessian of the cost function is ill-conditioned \citep{bottou2018optimization}, and even in the best case SGA converges no faster than sublinearly \citep{agarwal2009information, saad2009line, pelletier1998almost}. 
Significant effort in enhancing the plain SGA focuses on deriving adaptive learning step sizes; notable methods include Adam \citep{kingma2014adam}, AdaGrad \citep{duchi2011adaptive} and Adadelta \citep{zeiler2012adadelta}.
An alternative approach involves employing the natural gradient, which we will discuss in the following section.

%=================================================%
\section{Natural Gradient}\label{Natural Gradient}
%=================================================%
Let $\mathcal Q=\{q_\lambda(\theta):\lambda\in\mathcal M\subset\mathbb{R}^D\}$
be the set of VB approximating probability distributions parameterized by $\lambda$.
We denote by $d$ the dimension of the model parameter $\theta$, and by $D$ the dimension of the variational parameter $\lambda$. 
Gradient-based search for the optimal $\lambda$ relies on the concept of gradient whose definition depends upon the metric on $\mathcal M$.
It turns out that the regular Euclidean metric may not be appropriate for measuring the distance between two densities indexed by different variational parameters.
For instance, by adapting examples given in \cite{salimbeni2018natural}, the pair of two Gaussians $\mathcal N(0,0.1)$ and $\mathcal N(0,1.1)$ look significantly different from each other, compared to the pair  $\mathcal N(0,1000)$ and $\mathcal N(0,1001)$.
Both pairs have the same Euclidean distance, 
while their KL divergences highlight a significant difference: the first pair exhibits a KL divergence of 1.6, whereas the second pair has a KL divergence of $2.5\times 10^{-7}$. 

Now consider two variational parameters $\lambda$, $\lambda+\delta\lambda$ and the KL divergence $\KL(q_{\lambda}|| q_{\lambda+\delta\lambda})$.
From \cite{Tran:STCO2021,tan2021analytic}, it can be seen that 
%By Taylor's expansion, and noting that $\mathbb E_{q_\lambda}\left[ \nabla_{\lambda} \log q_{\lambda}(\theta) \right] = 0$,
\begin{align*}
\KL(q_{\lambda}|| q_{\lambda+\delta\lambda})
%&=\mathbb E_{q_\lambda}\big[\log q_\lambda(\theta)-\log q_{\lambda+\delta\lambda}(\theta)\big]\\
%&\approx \mathbb E_{q_\lambda}\Big[\log q_\lambda(\theta)-\big(\log q_\lambda(\theta) +(\delta\lambda)^\top \nabla_\lambda\log q_\lambda(\theta) %+\frac{1}{2}(\delta\lambda)^\top\nabla^2_\lambda \log q_\lambda(\theta)\delta\lambda\big)\Big]\\
&\approx\frac{1}{2}(\delta\lambda)^\top I_F(\lambda) \delta\lambda,
\end{align*}
where,
\begin{equation}\label{eq: Fisher matrix}
I_F(\lambda):=-\mathbb E_{q_\lambda}\Big[\nabla^2_\lambda \log q_\lambda(\theta)\Big] = \mathbb E_{q_\lambda}\Big[\nabla_\lambda\log q_\lambda(\theta) (\nabla_\lambda\log q_\lambda(\theta))^\top\Big],
\end{equation}
is the Fisher information matrix of $q_\lambda$.
This shows that the local KL divergence around the point $q_\lambda\in\mathcal Q$ is characterized by the Fisher matrix $I_F(\lambda)$.
Therefore, a suitable metric between $\lambda$ and $\lambda+\delta\lambda$
is the Fisher-Rao metric $(\delta\lambda)^\top I_F(\lambda)\delta\lambda$.
As a result, assuming the objective function $\LB$ is smooth enough and for $l>0$, if one considers the following optimization problem,
\begin{equation}\label{eq:nat grad opt}
\displaystyle\arg\max_{\delta\lambda: (\delta\lambda)^\top I_F(\lambda)\delta\lambda=l}\Big\{\nabla_\lambda\LB(\lambda)^\top\delta\lambda\Big\},
\end{equation}
then through the method of Lagrangian multipliers, the steepest ascent is 
\begin{equation}\label{eq: nat grad def}
\delta\lambda=\nabla_{\lambda}^{\text{nat}}\LB(\lambda):=I^{-1}_F(\lambda)\nabla_{\lambda}\LB(\lambda).
\end{equation}
\cite{amari1998natural} termed this the natural gradient and popularized it in machine learning. 

Using the natural gradient, the update in \eqref{SGD:lambdaUpdate} becomes
\begin{equation}\label{eq:natural gradient solution}
\lambda^{(k+1)}=\lambda^{(k)}+\tau_{k+1} I^{-1}_F(\lambda^{(k)})\nabla_{\lambda}\LB(\lambda^{(k)}).
\end{equation}
In the statistics literature, the steepest ascent in the form \eqref{eq:natural gradient solution} has been used  
for a long time and is often known as Fisher's scoring in the context of maximum likelihood estimation \citep[see, e.g.][]{longford1987fast}. The efficiency of the natural gradient over the Euclidean gradient  
has been well documented \citep{sato2001online,hoffman2013stochastic,tran2017variational,martens2020new,tan2021analytic}.
The natural gradient is invariant under parameterization \citep{martens2020new}, meaning it remains unchanged across different coordinate systems and is an intrinsic geometric object. This property makes it particularly suitable for use when the variational parameter space 
$\mathcal M$ is a Riemannian manifold \citep{Tran:STCO2021} as it is coordinate-free and leverages the underlying geometry of the space.

In the special case of Gaussian approximations 
with a full covariance matrix or a Cholesky-factor covariance matrix,
it is possible to obtain the inverse Fisher matrix in closed form \citep{tan2021analytic,magris2022exact}.
Beyond these limited cases, however, it is challenging to 
accurately compute the natural gradient.
The natural gradient method requires the {\it analytic} computation 
of the Fisher information matrix and its {\it inversion}.
Even if an analytical expression of the Fisher matrix is obtained,
computing its inverse has a complexity of $O(D^\kappa)$, with $2 < \kappa\leq 3$ depending on various algorithms.
It is therefore either analytically infeasible or prohibitively computationally expensive to use
natural gradient in many modern statistical applications; current
practice resorts to heuristic workarounds that can affect the results of Bayesian inference \citep{martens2020new,Lopatnikova:ICASSP}.

%=================================================%
\section{Inversion Free Natural Gradient Variational Bayes}\label{Inversion Free Natural Variational Bayes}
%=================================================%
This section first presents the approach for approximating the inverse of Fisher matrix. 
We then present the inversion free natural gradient Variational Bayes method, referred to as IFVB,
and its weighted averaged version AIFVB.
The IFVB and AIFVB methods are stochastic natural gradient descent algorithms that avoid computing the Fisher matrix and its inversion altogether.
As explained later, these methods also enable us to deal with
situations where the estimate of Fisher matrix has eigenvalues with significantly
different orders of magnitude (i.e., poor conditioning).
In this section and Section \ref{Convergence analysis}, we will denote $\mathcal{L}(\lambda)  = - \LB(\lambda)$, which can be viewed as the loss function,
and the problem of maximizing the lower bound becomes the minimization of $\mathcal{L}$.

\subsection{Recursive Estimation of $I^{-1}_{F}(\lambda)$}
%Bases on \eqref{eq: Fisher matrix}, one could approximate
%$I_F$ as follows,
%\begin{align}\label{FisherInformation-appproximation}
%I_F(\lambda)\approx \widehat{I_F}(\lambda):=
%\frac{1}{S}\sum_{j=1}^S \nabla_\lambda\log q_\lambda(\theta_j) (\nabla_\lambda\log q_\lambda(\theta_j))^\top,\;\;\theta_{1} , \ldots, \theta_{S} \sim q_{\lambda},
%\end{align}
%for some large sample size $S$.
%This approximation expects to approximate $I_F$ well as $S$ gets larger.
%However, the major drawback of the above approximation is that $\widehat{I_F}(\lambda)$ is not guaranteed to be positive definite and its inversion might not exist,
%which makes such an approximation unsuitable for the natural gradient descent algorithm.

For each $s=1,2\ldots,$ let 
\begin{equation}\label{Hiform}
  H_s= H_0+\sum_{j=1}^s\nabla_\lambda\log q_\lambda(\theta_j) (\nabla_\lambda\log q_\lambda(\theta_j))^\top,\;\;\;\theta_j\sim q_\lambda,
\end{equation}
where $H_0$ is some positive definite matrix, e.g, $ H_0=\epsilon\mathbb I_D$ with $\epsilon > 0$ and $\mathbb I_D$ the identity matrix of size $D$.  
Theorem \ref{thm:H} below says that $\mathbf{H}_s:=H_s/s$ is a consistent estimate of $I_{F}$ defined in \eqref{eq: Fisher matrix} as $s\to\infty$,
and provides a recursive procedure for obtaining $ H_s^{-1}$. This recursive procedure computes $ H_{s+1}^{-1}$
from $H_{s}^{-1}$ without resorting to the usual (expensive and error prone) matrix inversion.  Additionally,
the symmetry and positivity of $H_s$, and hence of $H_s^{-1}$, is preserved, which is an important  property.

\begin{thm}\label{thm:H}
Let $\phi_s=\nabla_\lambda\log q_\lambda(\theta_s)$. We have that
\begin{equation}\label{eq: IF recursive}
H_{s+1}^{-1}= H^{-1}_s-\left(1+\phi^\top_{s+1} H_{s}^{-1}\phi_{s+1}\right)^{-1}
 H_s^{-1}\phi_{s+1}\phi_{s+1}^\top H_{s}^{-1},\;\;s=0,1,...
\end{equation}
In particular, the positivity and symmetry of $H_s$ is
preserved for all $s$. Furthermore,
\begin{align*}
\mathbf{H}_s=\frac{1}{s}  H_s\xrightarrow[s \to + \infty]{a.s} I_F(\lambda)\;\;\;\text{and}\;\;\;\;\mathbf{H}_s^{-1}\xrightarrow[s \to + \infty]{a.s} I^{-1}_F(\lambda).
\end{align*}
\end{thm}
\noindent The proof of Theorem \ref{thm:H} can be found in Appendix \ref{Appendix:Prooof of Basic averaging theorem}. The expression \eqref{eq: IF recursive} suggests that $H_{s}^{-1}$ is a sum of outer products - a property that can be exploited to avoid storage of large matrices; see Remark \ref{re:remark 4.2}. 

\subsection{Inversion Free Natural Gradient Variational Bayes}
Theorem \ref{thm:H} suggests that one can approximate the inverse Fisher matrix $I^{-1}_{F}$ by $\mathbf{H}_s^{-1}=s H_s^{-1}$ for some large $s$,
where $ H_s^{-1}$ is calculated recursively as in \eqref{eq: IF recursive}.
This method does not require an analytic calculation of $I_{F}$ and its inversion; also, the estimate $\mathbf{H}_s^{-1}$ is guaranteed to be symmetric and positive definite.
However, a direct application of Theorem \ref{thm:H} for computing the natural gradient can be inefficient for two reasons. 

First, in order to ensure the consistency of estimates $\lambda^{(k)}$ from \eqref{eq:natural gradient solution},
where the inverse Fisher matrix is replaced by its estimate $\mathbf{H}_s^{-1}$,
one must control the eigenvalues of the estimate $\mathbf{H}_s^{-1}$. See \cite{bercu2020efficient} and \cite{boyer2023asymptotic} for related discussion in the context of stochastic Newton's method.
With this aim, we follow \cite{boyer2023asymptotic} and modify \eqref{Hiform} as follows
\begin{equation}\label{eq: H_s def}
A_s(\lambda) =H_0 + \sum_{j=1}^{s}\nabla_\lambda\log q_{\lambda }(\theta_{j}) (\nabla_\lambda\log q_{\lambda }(\theta_{j }))^\top + c_{\beta}\sum_{j=1}^{s}j^{-\beta}Z_{j}Z_{j}^\top,\;\;s=1,2,\ldots 
\end{equation}
where $\theta_j\sim q_\lambda(\cdot)$, $Z_{1} , \ldots , Z_{s}\sim \mathcal{N}(0,\mathbb{I}_D)$ are independent standard Gaussian vectors of dimension $D$, $c_{\beta}  \geq 0$ and $\beta \in (0, \alpha -1/2)$ for some $\alpha\in (1/2,1)$. 
%In addition, one can still update the inverse of $A_s$ using Riccati's formula twice \citep{cenac2020efficient}. Indeed, 
Theorem \ref{thm:H1} in the Appendix shows that $A_s/s$ converges almost surely to $I_F(\lambda)$, and that 
$A_{s+1}^{-1}$ can be  computed recursively  as follows
\begin{equation}\label{Riccati-twice-algorithm for A}
\begin{cases}
A_{s+ \frac{1}{2}}^{-1} = A_{s}^{-1} - \left( 1+ \phi_{s+1}^{\top} A_{s}^{-1} \phi_{s+1} \right)^{-1}A_{s}^{-1} \phi_{s+1}\phi_{s+1}^{\top} A^{-1}_{s}\\
A_{s+ 1}^{-1} = A_{s+ \frac{1}{2}}^{-1} - c_{\beta}(s+1)^{-\beta} \left( 1+ c_{\beta}(s+1)^{-\beta} Z_{s+1}^{\top}A_{s+ \frac{1}{2}}^{-1}Z_{s+1} \right)^{-1} A_{s+ \frac{1}{2}}^{-1}Z_{s+1}Z_{s+1}^{\top}A_{s+ \frac{1}{2}}^{-1}
\end{cases}
\end{equation}
with $\phi_s=\nabla_\lambda\log q_\lambda(\theta_s)$.
As being shown later in the proof of Theorem \ref{theo::as}, taking $c_{\beta} > 0$ ensures
the smallest eigenvalue of the Fisher matrix estimate not going to zero faster than $O(s^{-\beta})$ almost surely, hence enabling strongly consistent estimates.

Second, a direct use of \eqref{eq: H_s def}
would require a separate recursive procedure \eqref{Riccati-twice-algorithm for A} 
to calculate the inverse Fisher estimate $sA_s^{-1}(\lambda^{(k)})$ in each update $\lambda^{(k)}$, $k=1,2,...$, in \eqref{eq:natural gradient solution}.
This might make the VB training procedure in \eqref{eq:natural gradient solution} computationally expensive.
Instead, we propose a ``streamlined" version that updates the Fisher matrix estimate along with the iterates $\lambda^{(k)}$. To this end, we define
\begin{equation}\label{eq: H_s new def}
\tilde{H}_s=H_0 + \sum_{k=0}^{s-1}\nabla_\lambda\log q_{\lambda^{(k)} }(\theta_{k+1}) \big(\nabla_\lambda\log q_{\lambda^{(k)} }(\theta_{k+1})\big)^\top + c_{\beta}\sum_{k=1}^{s}k^{-\beta}
Z_{k}Z_{k}^\top,\;\;s=0,1,...,
\end{equation}
where $\theta_{k+1}\sim q_{\lambda^{(k)}}(\cdot)$. 
Note that, unlike \eqref{eq: H_s def},
$\tilde{H}_s$ in \eqref{eq: H_s new def} depends on the iterates $\lambda^{(k)}$, $k<s$. Similar to \eqref{Riccati-twice-algorithm for A},
$\tilde{H}_{s+1}^{-1}$ can be  computed recursively  as follows
\begin{equation}\label{Riccati-twice-algorithm}
\begin{cases}
\tilde{H}_{s+ \frac{1}{2}}^{-1} = \tilde{H}_{s}^{-1} - \left( 1+ \phi_{s+1}^{\top} \tilde{H}_{s}^{-1} \phi_{s+1} \right)^{-1}\tilde{H}_{s}^{-1} \phi_{s+1}\phi_{s+1}^{\top} \tilde{H}^{-1}_{s}\\
\tilde{H}_{s+ 1}^{-1} = \tilde{H}_{s+ \frac{1}{2}}^{-1} - c_{\beta}(s+1)^{-\beta} \left( 1+ c_{\beta}(s+1)^{-\beta} Z_{s+1}^{\top}\tilde{H}_{s+ \frac{1}{2}}^{-1}Z_{s+1} \right)^{-1} \tilde{H}_{s+ \frac{1}{2}}^{-1}Z_{s+1}Z_{s+1}^{\top}\tilde{H}_{s+ \frac{1}{2}}^{-1},
\end{cases}
\end{equation}
where $\phi_{s+1}=\nabla_\lambda\log q_{\lambda^{(s)} }(\theta_{s+1})$.

The Euclidean gradient of the loss function can be estimated as 
\begin{equation}\label{eq: LB gradient new}
\widehat{\nabla_{\lambda}\mathcal{L} } \left( \lambda \right) = -\frac{1}{B}\sum_{i=1}^{B} \nabla_{\lambda} \log q_{\lambda} \left( \theta_{i} \right) \times h_{\lambda} \left( \theta_{i} \right), 
\end{equation}
where the $\theta_{i}$'s are $B$ i.i.d samples from $q_{\lambda}$.

Putting these together, we propose the Inverse Free Natural Gradient VB algorithm, referred to as IFVB and outlined in Algorithm \ref{alg: IFVB ver 1}. 
The performance of this algorithm is found not sensitive to $\epsilon$ and $c_\beta$; both are set to 1 in our examples below unless stated otherwise. 
We found that the algorithm works well for $\beta$ in the range (0.1,0.5).

\begin{algorithm}[h!]
        \caption{Inversion-Free Natural Gradient Variational Bayes  (IFVB)}\label{alg: IFVB ver 1}
        \begin{algorithmic}[1]
            \REQUIRE Choose an initial value $\lambda^{(0)}$, $\epsilon>0$  and $c_\beta\geq0$. 
            \STATE $H_0=\epsilon\mathbb{I}_D$
            \FOR {$s=0,1,\ldots, $}
             \STATE Compute $\widehat{\nabla_{\lambda} \mathcal L}(\lambda^{(s)})$.
             \STATE Sample $\theta_{s+1}\sim q_{\lambda^{(s)}}(\theta)$, $Z_{s+1}\sim \mathcal{N}(0,\mathbb{I}_D)$ and let $\phi_{s+1}=\nabla_\lambda\log q_{\lambda^{(s)}}(\theta_{s+1})$. Calculate $\tilde{H}_{s+1}^{-1}$ as in \eqref{Riccati-twice-algorithm}, and ${\mathbf{\tilde{H}}}_{s+1}^{-1}=(s+1)\tilde{H}_{s+1}^{-1}$.             
             \STATE Update $\lambda^{(s+1)}  = \lambda^{(s)} - \tau_{s+1} {\mathbf{\tilde{H}}}_{s+1}^{-1} \widehat{\nabla_{\lambda}\mathcal{L} } \left( \lambda^{(s)} \right) $.             
             \ENDFOR
            \end{algorithmic}
\end{algorithm}

\begin{remark}
Some remarks are in order. Line $\#4$ in Algorithm \ref{alg: IFVB ver 1} 
updates $\tilde{H}_{s+1}^{-1}$ from $\tilde{H}_{s}^{-1}$ using only one sample from $q_{\lambda^{(s)}}$.
Depending on applications, however, it might be beneficial to use $S$ samples ($S>1$) to update $\tilde{H}_{s+1}^{-1}$. One then needs to run \eqref{Riccati-twice-algorithm} for $S$ times,
and compute ${\mathbf{\tilde{H}}}_{s+1}^{-1}=S(s+1)\tilde{H}_{s+1}^{-1}$.
The factor $S$ is not practically important as gradient clipping, that keeps the norm of the natural gradient ${\mathbf{\tilde{H}}}_{s+1}^{-1} \widehat{\nabla_{\lambda}\mathcal{L} } \left( \lambda^{(s)} \right) $ below some certain value, is often used in the SGD literature \cite[see, e.g.,][]{goodfellow2016deep}.   
\end{remark}

\begin{remark}\label{re:remark 4.2}
For applications such as deep learning where the size $D$ of variational parameter $\lambda$ is large, it is important to note that implementation of Algorithm \ref{alg: IFVB ver 1} does not require storage of the matrices $\tilde{H}_{s+1}^{-1}$.
Let us take $c_\beta=0$ to simplify the exposition; the extension to the case $c_\beta>0$ is straightforward. From \eqref{Riccati-twice-algorithm}, $\tilde{H}_{s+1}^{-1}$ takes the form
\begin{equation}\label{eq:outer product}
    \tilde{H}_{s+1}^{-1} = \frac{1}{\epsilon}\mathbb{I}_D-\sum_{k=1}^{s+1}\psi_k\psi_k^\top,\;\;\psi_k=\left( 1+ \phi_{k+1}^{\top} \tilde{H}_{k}^{-1} \phi_{k+1} \right)^{-1/2}\tilde{H}_{k}^{-1} \phi_{k+1},
\end{equation}
which is presented by outer products. As the result, all the required matrix-vector multiplications can be performed efficiently. That is,
\begin{equation*}
\begin{split}
\tilde{H}_{s+1}^{-1}\widehat{\nabla_{\lambda}\mathcal{L} } \left( \lambda^{(s)} \right) &= \frac{1}{\epsilon}\widehat{\nabla_{\lambda}\mathcal{L} } \left( \lambda^{(s)} \right)^\top\widehat{\nabla_{\lambda}\mathcal{L} } \left( \lambda^{(s)} \right)-\sum_{k=1}^{s+1}\big(\psi_k^\top\widehat{\nabla_{\lambda}\mathcal{L} } \left( \lambda^{(s)} \right)\big)\psi_k\\
\tilde{H}_{s}^{-1}\phi_{s+1} &= \frac{1}{\epsilon}\phi_{s+1}^\top\phi_{s+1}-\sum_{k=1}^{s}\big(\psi_k^\top\phi_{s+1}\big)\psi_k
\end{split}
\end{equation*}
which does not necessitate storage of the matrices $\tilde{H}_{s}^{-1}$. It is also natural to only keep the last $K$ outer products, for some $K\geq1$ ($K=100$ in our examples below), in \eqref{eq:outer product} to further reduce the computation
\begin{equation}\label{eq:outer product approx}
    \tilde{H}_{s+1}^{-1} \approx \frac{1}{\epsilon}\mathbb{I}_D-\sum_{k=s-K+2}^{s+1}\psi_k\psi_k^\top.
\end{equation}
The computational complexity of Algorithm \ref{alg: IFVB ver 1} is therefore $O(s\times D)$, with $s$ the number of iterations. This complexity is the same as that of the popular adaptive learning methods such as Adam and AdaGrad.
\end{remark}

\subsection{Weighted Averaged IFVB Algorithm}
It is a common practice in the SGD literature to use an average of the iterates $\lambda^{(k)}$ to form the final estimate of the optimal $\lambda^*$ \citep{goodfellow2016deep,boyer2023asymptotic}.
The weighted averaged estimate is of the form
\begin{equation*}
\overline{\lambda}^{(s+1)} = \frac{1}{\sum_{k=1}^{s+1}w_k}\sum_{k=1}^{s+1}w_k\lambda^{(k)} = \overline{\lambda}^{(s)}+\frac{w_{s+1}}{\sum_{k=1}^{s+1}w_k}\big(\lambda^{(s+1)}-\overline{\lambda}^{(s)}\big),    
\end{equation*}
where the weights $w_k>0$.
Inspired by \cite{boyer2023asymptotic}, we select $w_k=\big(\log(k)\big)^w$ with $w\geq0$.
This averaging scheme puts more weight on recent estimates $\lambda^{(k)}$ if $w>0$.
If one chooses $w=0$, this leads to the uniform averaging technique, i.e., 
\begin{equation*}
\overline{\lambda}^{(s)} =\frac{1}{s}\sum_{k=1}^s {\lambda}^{(k)}. 
\end{equation*}
We use $w=2$ in all the numerical examples in Section \ref{Numerical Example}.
We arrive at the Weighted Averaged IFVB Algorithm (AIFVB)
%\MNT{ Currently, we need to compute and store the matrix $\tilde{\mathbf{H}}_{s}^{-1}$ in each iteration,
%while what we really need is just the vector $\tilde{\mathbf{H}}_{s}^{-1} \widehat{\nabla_{\lambda}\mathcal{L} } \left( \lambda^{(s)} \right)$.
%It's expensive to store matrix $\tilde{\mathbf{H}}_{s}^{-1}$ in high dimensions.
%It'd be great if we can update the vector $\tilde{\mathbf{H}}_{s}^{-1} \widehat{\nabla_{\lambda}\mathcal{L} } \left( \lambda^{(s)} \right)$ directly without computing the matrix.
%Duy suggests that we could get some ideas from \cite{chen2022first}}
defined recursively for all $s \geq 0$ by
\begin{align*}
\lambda^{(s+1)} & = \lambda^{(s)} - \tau_{s+1} \tilde{\mathbf{H}}_{s+1}^{-1} \widehat{\nabla_{\lambda}\mathcal{L} } \left( \lambda^{(s)} \right) \\
\overline{\lambda}^{(s+1)} & = \overline{\lambda}^{(s)} + \frac{\big(\log (s+1)\big)^{w}}{\sum_{k=0}^{s} \big(\log (k+1)\big)^{w}} \left( \lambda^{(s+1)} - \overline{\lambda}^{(s)} \right).
\end{align*}

\noindent For the AIFVB algorithm, we consider the following estimate of the Fisher matrix 
%\blue{Did you change the index of $\tilde{\mathbf{H}}_{s} $ below ?}
\begin{equation}\label{eq: H_s def new}
\tilde{\mathbf{H}}_{s} = \frac{1}{s} \left(H_0 + \sum_{k=0}^{s-1}\nabla_\lambda\log q_{\overline{\lambda}^{(k )}}(\overline{\theta}_{k+1}) (\nabla_\lambda\log q_{\overline{\lambda}^{(k)}}(\overline{\theta}_{k+1}))^\top + c_{\beta}\sum_{k=1}^{s} k^{-\beta}Z_{k}Z_{k}^{T} \right),
\end{equation}
for $s =1,2,...$, 
where $\overline{\theta}_{k+1} \sim q_{\overline{\lambda}^{(k)}} (\theta)$, the $Z_{k}$'s are defined as before.
Similar to \eqref{Riccati-twice-algorithm}, $\tilde{\mathbf{H}}_{s+1}^{-1}$ can be updated recursively. 

\begin{remark}
In the estimate \eqref{eq: H_s def new}, we compute $\tilde{\mathbf{H}}_{s}$ using the averaged iterates $\overline{\lambda}^{(k)}$,
as this can lead to a faster convergence.
%As shown in Theorem \ref{theo::wasn} and confirmed in the numerical examples in Section \ref{Numerical Example}, this leads to a better converge rate.
Nonetheless, one can also use $\lambda^{(k)}$ instead of $\overline{\lambda}^{(k)}$ in \eqref{eq: H_s def new}.
\end{remark}

Putting these together, we have the Weighted Averaged IFVB Algorithm, outlined in Algorithm \ref{NewtonMethodVersion2}.

\begin{algorithm}
        \caption{ Weighted Averaged IFVB Algorithm (AIFVB)}\label{NewtonMethodVersion2}
        \begin{algorithmic}[1]
            \REQUIRE Choose an initial value $\lambda^{(0)}=\overline{\lambda}^{(0)}$, $\epsilon>0$ and $c_\beta\geq0$. 
            \STATE $H_0=\epsilon\mathbb{I}_D$
            \FOR {$s=0,1,\ldots, $}
             \STATE Compute $\widehat{\nabla_{\lambda} \mathcal L}(\lambda^{(s)})$.
             \STATE Sample $\overline{\theta}_{s+1}\sim q_{\overline{\lambda}^{(s)}}(\theta)$, $Z_{s+1}\sim \mathcal{N}(0,\mathbb{I}_D)$ and let $\phi_{s+1}=\nabla_\lambda\log q_{\overline{\lambda}^{(s)}}(\overline{\theta}_{s+1})$. Calculate $\tilde{H}_{s+1}^{-1}$ as in \eqref{Riccati-twice-algorithm}, and ${\mathbf{\tilde{H}}}_{s+1}^{-1}=(s+1)\tilde{H}_{s+1}^{-1}$.        
%             \STATE  Sample  and $\theta_{s+1,i}\sim q_{ {\lambda}^{(s)}}(\theta)$ for $i=1,\ldots ,B$, compute $\phi_{s+1}=\nabla_\lambda \log q_{\overline{\lambda}^{(s)}}(\theta_{s+1})$ and $\widehat{\nabla_{\lambda} \mathcal L}(\lambda^{(s)})$
%             \STATE Calculate $\tilde{H}_{s+1}^{-1}$ as in \eqref{Riccati-twice-algorithm}, and ${\mathbf{\tilde{H}}}_{s+1}^{-1}=(s+1)\tilde{H}_{s+1}^{-1}$             
             \STATE Update $\lambda^{(s+1)}  = \lambda^{(s)} - \tau_{s+1} \tilde{\mathbf{H}}_{s+1}^{-1} \widehat{\nabla_{\lambda}\mathcal{L} } \left( \lambda^{(s)} \right) $.
             \STATE Update $\overline{\lambda}^{(s+1)}  = \overline{\lambda}^{(s)} + \frac{\big(\log (s+1)\big)^{w}}{\sum_{k=0}^{s} \big(\log (k+1)\big)^{w}} \left( \lambda^{(s+1)} - \overline{\lambda}^{(s)} \right) $.    
             \ENDFOR
            \end{algorithmic}
\end{algorithm}

%\blue{
%\begin{rem}
%There are several directions that we add to make the paper stronger:
%\begin{enumerate}
%\item Since the likelihood function is usually intractable:
%$\log(p(y|\theta))=\sum_{i=1}^n \log p(y_i|\theta)$ for large
%$n$, which is usually often the case in the regime of big data.
%We can follow the idea of \cite{rhee2015unbiased} to design
%an unbiased estimate for $\log(p(y|\theta)$.
%That is, we will replace $\log(p(y|\theta)$
%by 
%$$
%\widehat{\log(p(y|\theta)}=\sum_{n=1}^N \frac {\mathbf{1}_{N\geq n}}{\mathbb P(N\geq n)} (\log(p(y_{n+1}|\theta)-\log p(y_n|\theta))\quad N\sim f_N.
%$$
%\item We can use batch estimation
%to estimate the gradient $\nabla_{\lambda}\mathcal L(\lambda)$
%as done in \cite{fujisawa2021multilevel}.
%\end{enumerate}
%\end{rem}
%}
%\subsection{Example}\label{Example:GaussianApproximation}
%We provide several examples
%to illustrate the idea
%and these examples will 
%be reconsidered in the numerical example
%section later.\\
%
%\noindent\textbf{Example 1.}

\section{Convergence analysis}\label{Convergence analysis}
For the convergence analysis in this section, we consider a step size of the form $\tau_{k} = \frac{c_{\alpha}}{\left( c_{\alpha}' + k \right)^{\alpha}}$ with $c_{\alpha} >0$, $c_{\alpha}'\geq 0$ and $\alpha \in (1/2,1)$. 
Write
\begin{align*}
\nabla_\lambda \mathcal L(\lambda)=\mathbb E_{q_\lambda}[ - \nabla_\lambda\log q_\lambda(\theta)\times  h_\lambda(\theta)] =: \mathbb E_{q_\lambda}[\ell (\theta,\lambda)].
\end{align*}

\noindent It can be seen that\footnote{See Appendix \ref{app: appendix on Hessian} for detailed calculations; see also  \cite{tang2019variational,tan2021analytic}.
}
\[
\nabla_\lambda^2\mathcal L(\lambda) = I_F(\lambda) + \int\nabla_\lambda^2 q_\lambda(\theta)\big(\log q_\lambda(\theta)-\log p(\theta|y) \big)d\theta .
\]
At the optimal $\lambda=\lambda^*$, $q_{\lambda^*}(\theta)\approx p(\theta|y)$, the second term above expects to be close to zero.
We can therefore expect that, in a neighbourhood of $\lambda^*$, the Fisher $I_F(\lambda)$ behaves like the Hessian $\nabla_\lambda^2\mathcal L(\lambda)$.
This motivates us to adapt the results from stochastic Newton algorithms \cite[see, e.g.,][]{bercu2020efficient,boyer2023asymptotic}
for convergence analysis of IFVB and AIFVB proposed in this paper. 

We now provide the convergence analysis for the AIFVB algorithm; a minor modification provides convergence results for IFVB.
%The first convergence result of the iterates $\{\lambda^{(k)},k=1,2,...\}$ in Theorem \ref{theo::as} requires convexity of $\mathcal{L}$.
%For the rest of the results in this section, under the convergence assumption of the estimates $\{\lambda^{(k)},k=1,2,...\}$,
%we obtain the convergence rate and a central limit theorem.

\begin{thm}\label{theo::as}
Assume that $\mathcal{L}$ is  twice differentiable and that there is $L_{0}$ such that for all $\lambda$, $\left\| \nabla_{\lambda}^{2} \mathcal{L} (\lambda) \right\|_{op} \leq L_{0}$. Suppose that  $\nabla_\lambda\mathcal L(\lambda^*)=0$ and $c_{\beta} > 0$. Assume also that there are non negative constants $C_{0},C_{1}$ such that  for all $\lambda$
\[
\mathbb E\left[ \norm{\ell(\theta,\lambda)}^2\right]\leq C_0+C_1(\mathcal L(\lambda)-\mathcal L(\lambda^*))
\]
%and that there are positive constants $C_{0}',C_{1}'$ such that for all $\lambda$, 
%\begin{equation}\label{eq:fourth moment assumption}
%\mathbb{E}\left[ \left\| \nabla_{\lambda} \log q_{\lambda} (\theta) \right\|^{4} \right] \leq C_{0}' + C_{1}' \left( \mathcal{L}(\lambda) - \mathcal{L}\left( \lambda^{*} \right)\right)^{2} .   
%\end{equation}
and that there are positive constants $C_{0}',C_{1}'$ such that for all $\lambda$, 
\begin{equation}\label{eq:fourth moment assumption}
\mathbb{E}\left[ \left\| \nabla_{\lambda} \log q_{\lambda} (\theta) \right\|^{4} \right] \leq C_{0}' + C_{1}' \left( \mathcal{L}(\lambda)- \mathcal{L}(\lambda^{*}) \right)^2 .   
\end{equation}
Suppose that $C_{1}'=0$ or that $\mathcal{L}$ is convex. Then the estimates $\lambda^{(k)}$, $\overline{\lambda}^{(k)}$ $k=0,1,...$, from Algorithm \ref{NewtonMethodVersion2} satisfy:
\begin{enumerate}
    \item[(i)] $\mathcal{L}\left( \lambda^{(k)} \right) - \mathcal{L}\left( \lambda^{*} \right)$ converges almost surely to a finite random variable. 
    \item[(ii)] $\min_{k=0}^{s}\left\| \nabla \mathcal{L}\left( \lambda^{(k)} \right) \right\|^{2} = o \left( s^{-(1-\alpha)} \right) $ a.s.
    \item[(iii)] $\min_{k=0}^{s}\left\| \nabla \mathcal{L}\left( \overline{\lambda}^{(k)} \right) \right\|^{2} = o \left( s^{-(1-\alpha)} \right) $ a.s.
\end{enumerate}
\end{thm}
If the convexity of $\mathcal{L}$ is satisfied, conclusions (ii) and (iii) in Theorem \ref{theo::as} imply the almost sure convergence of $\lambda^{(k)}$ and $\overline{\lambda}^{(k)}$ to $\lambda^*$.
The assumption in \eqref{eq:fourth moment assumption} might appear to look irrelevant, 
as the right-hand side term is model-dependent, i.e. depending on the prior and likelihood, while the left-hand side is not.
One can simply replace \eqref{eq:fourth moment assumption} with an assumption on the uniform bound of the fourth order moment of $\nabla_\lambda\log q_{\lambda}$,
which obviously implies \eqref{eq:fourth moment assumption}. We use \eqref{eq:fourth moment assumption} to keep the result as general as possible.
%Observe that a main difference with the work of \cite{cenac2020efficient} is that we do not need to assume the uniform bound of the fourth order moment of the gradient of $\log %q_{\lambda}$.\MNT{As discussed, I've updated condition \eqref{eq:fourth moment assumption} that the fourth order moment of $\nabla\log q_{\lambda}$ is bounded.
%Antoine, can you please update the proof accordingly?} \textcolor{purple}{I can easily modify the proof accordingly, but my interrogation is why we do not keep the previous bound (with a term %depending on the loss?)}

We now give the consistency of the estimates of the Fisher information given in \eqref{eq: H_s def new}.
\begin{cor}\label{cor::as::hs}
Suppose that $\lambda^{(k)}$ converges almost surely to $\lambda^{*}$, and that the map $\lambda \longmapsto I_{F}(\lambda)$ is continuous at $\lambda^{*}$. Suppose also that \eqref{eq:fourth moment assumption} is satisfied.
%there are positive constants $C_{0}',C_{1}'$ such that for all $\lambda$, 
%\[
%\mathbb{E}\left[ \left\| \nabla_{\lambda} \log q_{\lambda} (\theta) \right\|^{4} \right] \leq C_{0}' + C_{1}' \left( \mathcal{L}(\lambda) - \mathcal{L}\left( \lambda^{*} %\right)\right)^{2} .
%\]
Then
\[
\tilde{\mathbf{H}}_{s} \xrightarrow[s\to +\infty]{a.s} I_{F} \left( \lambda^{*} \right).
\]
\end{cor}
Note that we can obtain the convergence rate of $\tilde{\mathbf{H}}_{s}$ to $I_{F}(\lambda^{*})$,
or even the rate of convergence of $\tilde{\mathbf{H}}_{s}$ to the Hessian $\nabla_\lambda^2 \mathcal L(\lambda^*)$.
Please see Appendix \ref{Appendix:Convergence rate to Hessian} for more details.
 
Without the convexity assumption, the proof of Theorem \ref{theo::as} implies that 
the estimates $\lambda^{(k)}$ converge to a stationary point $\lambda^*$ of the objective $\mathcal L$, i.e. $\nabla_\lambda\mathcal L(\lambda^*)=0$.
The result below gives a convergence rate of the estimates $\lambda^{(k)}$ to such $\lambda^*$.
\begin{thm}\label{theo::rate::lambda}
Suppose that $\lambda^{(k)}$ converges almost surely to $\lambda^{*}$. Assume that the functional $\mathcal{L}$ is differentiable with $\nabla_\lambda\mathcal L(\lambda^*)=0$ and twice continuously differentiable in a neighborhood of $\lambda^{*}$. Suppose also that
there  are $\eta > \frac{1}{\alpha}-1$ and positive constants $C_{\eta,0},C_{\eta,1}$ such that for all $\lambda$,
\[
\mathbb{E}\left[ \left\| \ell(\theta , \lambda) \right\|^{2+2\eta} \right] \leq C_{\eta,0} + C_{\eta,1}\left( \mathcal{L} (\lambda) - \mathcal{L}\left( \lambda^{*} \right) \right)^{1+\eta} 
\]
and that inequality \eqref{eq:fourth moment assumption} holds. Suppose also that $\nabla_{\lambda}^{2}\mathcal{L} \left( \lambda^{*} \right)$ and $I_{F} \left( \lambda^{*} \right)$ are positive. Then 
\begin{equation}\label{eq:convergence rate 1}
\norm{\lambda^{(s)}-\lambda^*}^2=\mathcal O\left(\frac{\log s}{s^\alpha}\right) \quad a.s. 
\end{equation}
\end{thm}
Observe that \eqref{eq:convergence rate 1} is the usual rate of convergence for (adaptive) stochastic gradient type algorithms \citep{bercu2020efficient,nguyen2021unified}.
The following theorem shows that the weighted averaged estimates achieve a better convergence rate and are asymptotically efficient.
\begin{thm}\label{theo::wasn}
Suppose that the assumptions in Theorem \ref{theo::rate::lambda} hold.
Assume that the functional 
\[
\lambda \longmapsto \Sigma({\lambda}):=\mathbb{E}_{q_\lambda} \left[ \ell (\theta, \lambda) \ell(\theta , \lambda)^\top \right]
\]
is continuous at $\lambda^{*}$, and there are a neighborhood $V^{*}$ of $\lambda^{*}$ and a constant $L_{\delta}$ such that for all $\lambda \in V^{*}$,
\begin{equation}\label{eq::delta}
\left\| \nabla_{\lambda} \mathcal{L}(\lambda) - \nabla^{2}\mathcal{L} (\lambda^*)(\lambda - \lambda^*) \right\| \leq L_{\delta}\left\| \lambda - \lambda^* \right\|^{2}. 
\end{equation}
Then 
\[
\norm{\overline{\lambda}^{(s)}-\lambda^*}^2=\mathcal O\left(\frac{\log s}{s}\right) \quad a.s. 
\]
In addition, 
\[
\sqrt{Bs} \left( \overline{\lambda}^{(s)} - \lambda^{*} \right) \xrightarrow[n\to + \infty]{law} \mathcal{N}\left( 0 , {\nabla_{\lambda}^{2}\mathcal{L}\left( \lambda^{*} \right)^{-1}\Sigma \left( \lambda^{*} \right)\nabla_{\lambda}^{2}\mathcal{L}\left( \lambda^{*} \right)^{-1}} \right) .
\]
\end{thm}

\begin{rmk}
    
Comparing Theorem \ref{theo::rate::lambda} and Theorem \ref{theo::wasn}, it is intriguing to observe that the averaging technique clearly contributes to an improvement in the rate of convergence.  Moreover, observe that $Bs$ represent the total number of samples generated for the AIFVB algorithm (without taking into account the ones for estimating the inverse of the Fisher information). In addition, following \cite{GBW2023}, i.e  taking $B=D$ and generating only one sample to estimate the inverse of the Fisher information, it leads to an algorithm with $O(sBD)$ operations, which is the same computational complexity as stochastic gradient type algorithms. To be more precise, taking $B=D$, and denoting by $N=sB$ the total number of simulated data used for estimating the gradients at each step of the algorithm, one has
\[
\sqrt{N} \left(\overline{\lambda}^{N/D} - \lambda^{*} \right) \xrightarrow[N\to + \infty]{\text{law}} \mathcal{N}\left( 0 , {\nabla_{\lambda}^{2}\mathcal{L}\left( \lambda^{*} \right)^{-1}\Sigma \left( \lambda^{*} \right)\nabla_{\lambda}^{2}\mathcal{L}\left( \lambda^{*} \right)^{-1}} \right) ,
\]
with $O(ND)$ operations, which is exactly the same rate and complexity as averaged stochastic gradient algorithms.

\end{rmk}

\begin{rmk}
To the best of our knowledge, this paper is the first to give a central limit theorem for the estimate of variational parameters in VB.
This result can have some interesting implications.
Given a variational family $\mathcal Q$ and data $y$, one can define the ``final" variational approximation of the posterior distribution $p(\theta|y)$ as
\begin{equation}\label{eq: final var approx}
q(\theta|y,\mathcal Q) = \int q_\lambda(\theta)p(\lambda|y)d\lambda
\end{equation}
with $p(\lambda|y)$ approximated by $\mathcal{N}\left( \lambda^* , \frac{1}{Bs}{\nabla_{\lambda}^{2}\mathcal{L}\left( \lambda^{*} \right)^{-1}\Sigma \left( \lambda^{*} \right)\nabla_{\lambda}^{2}\mathcal{L}\left( \lambda^{*} \right)^{-1}} \right)$, who in turn can be further approximated by $\mathcal{N}\left(\lambda^* , \frac{1}{Bs} {I_F\left( \lambda^{*} \right)^{-1}\Sigma \left( \lambda^{*} \right)I_F\left( \lambda^{*} \right)^{-1}} \right)$.
The posterior approximation in \eqref{eq: final var approx} not only takes into account 
the uncertainty in the SGD training procedure of $\lambda$, but also enlarges the variance in $q_{\lambda^*}(\theta)$.
We recall that underestimating the posterior variance is a well perceived problem in the VB literature \citep{blei2017variational}. 
\end{rmk}

%\section{Subsampling Inversion-Free Variational Bayes}

\section{Numerical Examples}\label{Numerical Example}
This section provides a range of examples
to demonstrate
the applicability of the inversion free
variational Bayes methods. 
In Sections \ref{Inversion Free Natural Variational Bayes} and \ref{Convergence analysis}, to be consistent with the optimization literature, we presented the VB problem in terms of optimizing the negative lower bound $\mathcal{L}(\lambda)$. In this section, to be consistent with the VB literature, we present the numerical examples in terms of maximizing the lower bound $\text{LB}(\lambda)$.
The five examples encompass a diverse array of domains
ranging from Gaussian approximation to normalizing flow Variational Bayes.
We provide four examples in this section; another example is provided in the Supplementary Material.
%The Github link to the implementation code is available in the unblinded version. 
The implementation code is available at \url{https://github.com/VBayesLab/Inversion-free-VB}. 

\noindent\textbf{Example 1 (Beta variational approximation).}
We consider an example in
which the posterior distribution and the natural gradient are given in closed-form.
This helps facilitate the comparison among the considered approximation approaches.
Following \cite{tran2017variational}, we generated $n=200$ observations $y_1,y_2,\ldots,y_{200}\sim\Ber(1,\theta=0.3)$, and let $\kappa=\sum_{i=1}^n y_i=57$.
The prior distribution of $\theta$ is chosen to be the uniform distribution on $(0,1)$.
The posterior distribution is $\Beta(\kappa+1,n-\kappa+1)$.
Let $\lambda^*=(\kappa+1,n-\kappa+1)^\top$. The variational distribution $q_\lambda(\theta)$ is chosen to be $\Beta(\alpha,\beta)$,
which belongs to the exponential family with the natural parameter $\lambda=(\alpha,\beta)^\top$.
We have 
\begin{align*}
\log q_\lambda(\theta)=\log\Gamma(\alpha+\beta)-\log\Gamma(\alpha)-\log\Gamma(\beta)
+(\alpha-1)\log\theta+(\beta-1)\log(1-\theta).
\end{align*}
Hence
\begin{align*}
 \nabla_\lambda\log q_\lambda(\theta)=
 \left(\psi(\alpha+\beta)-\psi(\alpha)+\log(\theta),
 \psi(\alpha+\beta)-\psi(\beta)+\log(1-\theta)   \right)^\top,
 \end{align*}
 where $\psi$ is the digamma function. In this case, the Fisher information matrix 
is available in a closed-form
\begin{align*}
I_F(\lambda)=\begin{pmatrix}
\varphi_1(\alpha)-\varphi_1(\alpha+\beta)& -\varphi_1(\alpha+\beta)\\
-\varphi_1(\alpha+\beta) & \varphi_1(\beta)-\varphi_1(\alpha+\beta)
\end{pmatrix},
\end{align*}
where $\varphi_1(x)$ is the trigamma function.
In addition, we have that
\begin{align*}
\nabla_\lambda\LB(\lambda)=\begin{pmatrix}
    (\kappa+1-\alpha)\big(\varphi_1(\alpha)-\varphi_1(\alpha+\beta)\big)-(n-\kappa+1-\beta)\varphi_1(\alpha+\beta)\\
    (n-\kappa+1-\beta)\big(\varphi_1(\beta)-\varphi_1(\alpha+\beta)\big)-(\kappa+1-\alpha)\varphi_1(\alpha+\beta)    
\end{pmatrix}.
\end{align*}
The Euclidean gradient ascent update is given by
\begin{align*}
\lambda^{(k+1)}=\lambda^{(k)}+\tau_k\nabla_\lambda\LB(\lambda^{(k)}),
\end{align*}
and the natural gradient ascent update is given by
\begin{align*}
\lambda^{(k+1)}&=\lambda^{(k)}+\tau_k I_F^{-1}(\lambda^{(k)})\nabla_\lambda\LB(\lambda^{(k)}).
\end{align*}

\begin{comment}
In addition, we have
\begin{align*}
\nabla_\lambda\KL(\lambda)=I_F(\lambda)\lambda-H(\lambda),
\end{align*}
with
\begin{align*}
H(\lambda)=(\kappa\varphi_1(\alpha)-n\varphi_1(\alpha+\beta), (n-\kappa)\varphi_1(\beta)-n\varphi_1(\alpha+\beta))^\top.
\end{align*}
The gradient ascent update is given by
\begin{align*}
\lambda^{(k+1)}&=\lambda^{(k)}+\tau_k\nabla_\lambda\KL(\lambda^{(k)}\\
&=\lambda^{(k)}+\tau_k\left( I_F(\lambda^{(k)})\lambda^{(k)}-H(\lambda^{(k)})\right).
\end{align*}
The natural gradient ascent update is given by
\begin{align*}
\lambda^{(k+1)}&=\lambda^{(k)}+\tau_k I_F^{-1}(\lambda^{(k)})\left( I_F(\lambda^{(k)})\lambda^{(k)}-H(\lambda^{(k)})\right).
\end{align*}
    
\end{comment}

Figure \ref{Comparing exact posterior density versus those obtained by gradient ascent, NGVB and IFVB}
compares the true posterior density
with those obtained by the Euclidean gradient
ascent, the (exact) natural gradient VB algorithm (NGVB) and IFVB,
using two initializations for $\lambda$: $\lambda^{(0)}=(5;45)^\top$ (Figure \ref{Comparing exact posterior density versus those obtained by gradient ascent, NGVB and IFVB}, left) and $\lambda^{(0)}=(25;25)^\top$ (Figure \ref{Comparing exact posterior density versus those obtained by gradient ascent, NGVB and IFVB}, right). We set $c_\beta=0$ in this example.
The step size for IFVB is set as $\tau_k=10/(1+k)^\alpha$ with $\alpha=0.6$ (see Section \ref{Convergence analysis}), while the step size for other methods is set to $1/(1+k)$.
All the algorithms used the same stopping rule, which terminates the training if the difference in $l_2$-norm of two successive updates is less than a tolerance of $10^{-5}$.
The figure shows that 
the posterior densities obtained by NGVB and IFVB
are superior to using the traditional Euclidean gradient. Furthermore,
these algorithms are insensitive to the initial $\lambda^{(0)}$.

\begin{figure}[h]
	\centering
	\includegraphics[scale=0.3]{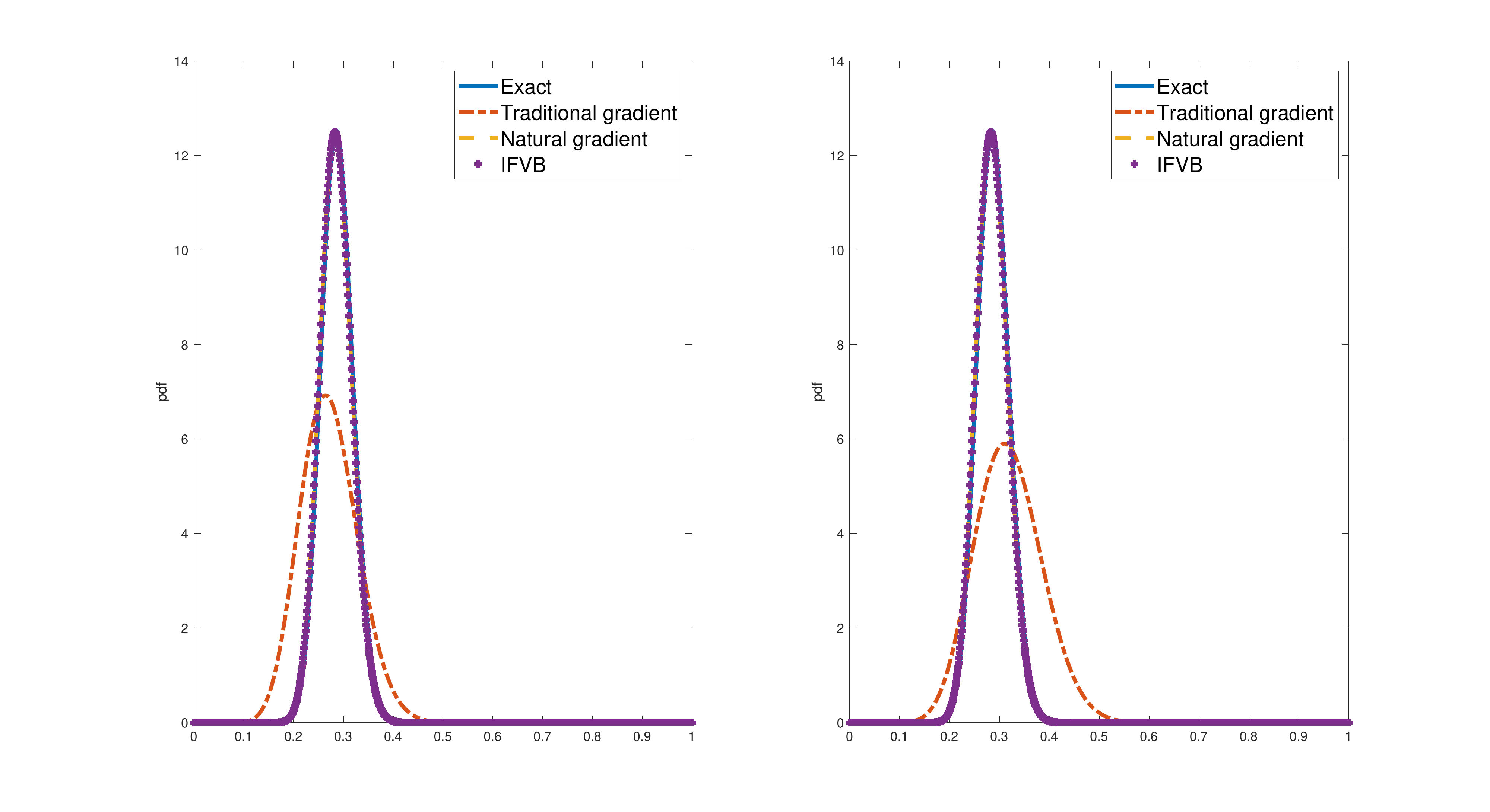}
\caption{Comparing the true posterior density versus those obtained by Euclidean gradient ascent, NGVB and IFVB.}\label{Comparing exact posterior density versus those obtained by gradient ascent, NGVB and IFVB}	
\end{figure}

As we know the exact posterior distribution in this case, we now further compare the performance of NGVB, IFVB and AIFVB using the initialization $\lambda^{(0)}=(5;45)^\top$.
Figure \ref{Comparing exact posterior density versus those obtained by gradient ascent, NGVB, IFVB and AIFVB} plots the error $\|\lambda^{(s)}-\lambda^*\|$ ($\|\bar{\lambda}^{(s)}-\lambda^*\|$ for AIFVB)
as the function of the number of iterations. 
It is clear from Figure \ref{Comparing exact posterior density versus those obtained by gradient ascent, NGVB, IFVB and AIFVB}
that both the IFVB and AIFVB algorithms require more steps to achieve a (relatively) equivalent error level compared to the exact natural gradient ascent algorithm. This observation aligns with the fact that both algorithms rely on an approximation of the inverse Fisher information matrix. As time advances, the quality of this approximation improves. Furthermore, it is evident that the incorporation of an averaging technique significantly aids in minimizing the error associated with the approximation.

\begin{figure}[h]
	\centering
	\includegraphics[scale=0.3]{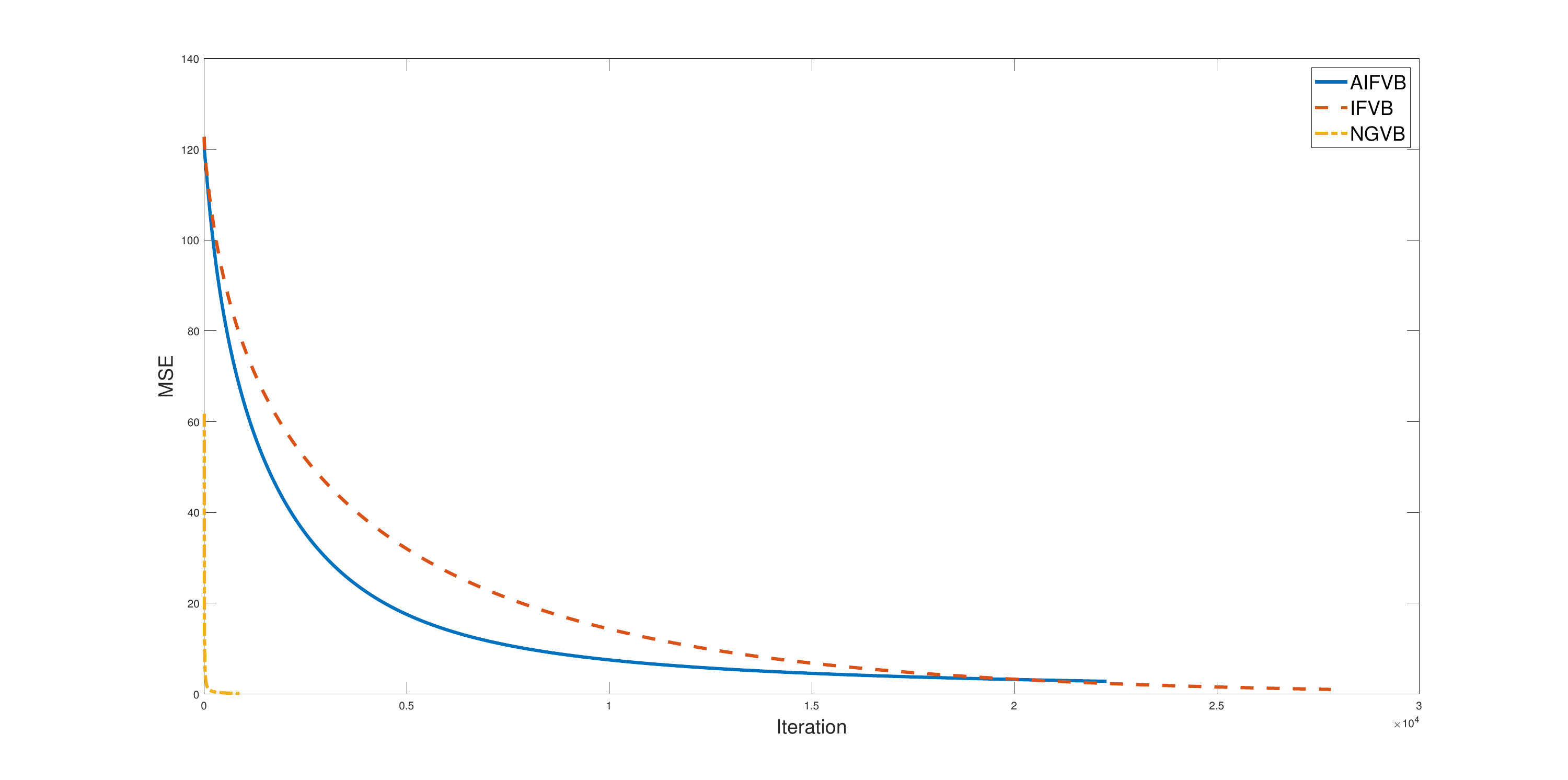}
\caption{Error versus iterations}\label{Comparing exact posterior density versus those obtained by gradient ascent, NGVB, IFVB and AIFVB}
\end{figure}

%\begin{figure}[h]
%	\centering
%	\includegraphics[scale=0.4]{Poisson_Compare_SGD_NG_IFVB_AIFVBV3.pdf}
%\caption{Error versus iterations}\label{Comparing exact posterior density versus those obtained by gradient ascent, NGVB, IFVB and AIFVB}
%\end{figure}

\vspace{1.5cm}
\noindent\textbf{Example 2 (Gaussian approximation).} 
We consider a Gaussian approximation example
which was considered in \cite{tan2021analytic}.
Specifically, it is assumed that $q_\lambda(\theta)=\mathcal N (\theta|\mu,\Sigma)$
which belongs to the exponential family and can be  written as
\begin{align*}
q_\lambda(\theta)=\exp\left(s(\theta)^\top\lambda -a(\lambda) \right),
\end{align*}
where $s(\theta)=(\theta^\top,\vech(\theta\theta^\top)^\top)^\top$
is the sufficient statistics,
$\lambda=(\mu^\top\Sigma^{-1},-\frac{1}{2}\vec(\Sigma^{-1})^\top \Gamma)^\top$
is the natural parameter, with $\Gamma$ the duplication matrix.
It can be seen that $a(\lambda)=\frac{1}{2}\mu^\top\Sigma^{-1}\mu+\frac{1}{2}\log|\Sigma|
+\frac{d}{2}\log(2\pi)$ is the log-partition function.
This implies the natural gradient
%\MNT{@Duy: We reserved $\mathcal L$ for MINUS lower bound before. That is, your $\mathcal L$ below should be $\LB$}
\begin{align*}
\nabla^{nat}_\lambda\LB=I^{-1}_F(\lambda)\nabla_\lambda\LB(\lambda)=I^{-1}_F(\lambda)I_F(\lambda)\nabla_m\LB(m)=\nabla_m\LB=
\begin{pmatrix}
\nabla_\mu\LB-2(\nabla_\Sigma\LB)\mu\\
\Gamma^\top\vec(\nabla_\Sigma\LB)
\end{pmatrix},
\end{align*}
where $\vec(\nabla_\Sigma\LB)=\nabla_{\vec(\Sigma)}\LB$.
From here \cite{tan2021analytic} derives the updates for $\lambda$ as follows,
\begin{enumerate}
\item $\Sigma^{-1}\leftarrow \Sigma^{-1} -2 \tau_k\nabla_\Sigma\LB$
\item $\mu\leftarrow \mu+\tau_k\Sigma\nabla_\mu\LB$.
\end{enumerate}
Obviously, the traditional (Euclidean) gradient ascent  is given by
\begin{enumerate}
\item $\mu\leftarrow \mu+\tau_k\nabla_\mu\LB$
\item $\Sigma \leftarrow \Sigma +\tau_k\nabla_\Sigma\LB$.
\end{enumerate}
Let us consider a concrete example: Consider the loglinear model for counts, $y_i\sim\Pois(\delta_i), 
\delta_i=\exp(x_i^\top\theta)$ where $x_i,\theta\in\mathbb R^d$ are the
vector of covariates and regression coefficients, respectively.
Assume the prior of $\theta$ is $p(\theta)\sim\mathcal N(0,\sigma_0^2 \mathbb{I}_d)$.
We use a Gaussian $q_\lambda(\theta)=\mathcal N(\theta|\mu,\Sigma)$
to approximate the true posterior distribution of $\theta$. 
Let $y=y_{1:n}$ and $X=(x_1,x_2,\ldots,x_n)^\top$.
For each $i=1,\ldots,n$ let $w_i=\exp(x_i^\top\mu+\frac{1}{2}x_i^\top\Sigma x_i)$
and $W=\diag(w_1,\ldots,w_n)$.
By computing directly \cite{tan2021analytic} showed that the closed form for the ELBO is,
\begin{align*}
\LB(\lambda)=y^\top X\mu -\sum_{i=1}^n(w_i+\log(y_i!))-\frac{\mu^\top\mu+Tr(\Sigma)}{2\sigma_0^2}
+\frac{1}{2}\log|\Sigma|+\frac{d}{2}(1-\log(\sigma_0^2)).
\end{align*}

From here it is immediate to see that  $\nabla_\mu\LB$, $\nabla_\Sigma\LB$, and $\nabla_{\vec(\Sigma)}\LB$ are given by 
\begin{align*}
&\nabla_\mu\LB=X^\top(y-w)-\frac{\mu}{\sigma_0^2},\\
&\nabla_\Sigma\LB=-\frac{1}{2}X^\top W X-\frac{1}{2\sigma_0^2}\mathbb I +\frac{1}{2}\Sigma^{-1}\\
&\nabla_{\vec(\Sigma)}\LB=\frac{1}{2}\vec(\Sigma^{-1}-\frac{1}{\sigma_0^2}\mathbb I-X^\top W X).
\end{align*}
Also we have
\begin{align*}
\log q_\lambda(\theta)=-\frac{d}{2}\log(2\pi)-\frac{1}{2}\log|\Sigma|-\frac{1}{2}(\theta-\mu)^\top\Sigma^{-1}(\theta-\mu).
\end{align*}
Hence 
\begin{align*}
\nabla_\lambda q_\lambda(\theta)&=(\nabla_\mu q_\lambda(\theta),\nabla_{\vec(\Sigma)}q_\lambda(\theta))\\
&=\left( (\theta-\mu)^\top\Sigma^{-1},\vec\left(-\frac{1}{2}\Sigma^{-1}+\frac{1}{2}\Sigma^{-1}(\theta-\mu)(\theta-\mu)^\top\Sigma^{-1}\right)^\top \right)^\top\in\mathbb R^{d+d^2}.
\end{align*}
%Therefore 
%\begin{align*}
%\nabla_\lambda q_\lambda(\theta)^\top \nabla_\lambda q_\lambda(\theta)=
%\end{align*}
%and
%\begin{align*}
%\nabla_\lambda q_\lambda(\theta) \nabla_\lambda q_\lambda(\theta)^\top=
%\end{align*}
%\red{Goal: get data and plot the ELBO from three approaches: Traditional gradient, Natural gradient, and IFVB}
%\red{There is (could be wrong or correct) formula for $\frac{\partial d_\lambda}{\partial \Sigma}$}
%\begin{align*}
%\frac{\partial q_\lambda}{\partial\Sigma}
%=2(\left(-\frac{1}{2}\Sigma^{-1}+\frac{1}{2}\Sigma^{-1}(\theta-\mu)(\theta-\mu)^\top\Sigma^{-1}  \right)
%-diag\left(-\frac{1}{2}\Sigma^{-1}+\frac{1}{2}\Sigma^{-1}(\theta-\mu)(\theta-\mu)^\top\Sigma^{-1}  \right).
%\end{align*}
%We may be test them both to see which one gives better answers.

We choose $n=200, d=3,\sigma_0^2=100,\theta=(1,\ldots,1), x_{ij}\sim\mathcal N(0,1),$ and simulate $y_1,\ldots,y_{200}$
from those parameters. In this case $\lambda\in \mathbb R^{12}$ hence $I_F\in\mathbb R^{12\times 12}$. The initial guesses are $\mu^{(0)}=(0,\ldots,0)^\top$ and
$\Sigma^{(0)}=10^{-2}\mathbb I_d$. Moreover, we choose the step size $1/(1000+k)^{0.75} $ and $ c_\beta=1$. As we know
the analytical form of the ELBO in this example,
in Figure \ref{Gaussian Approximation: Comparing NGVB versus IFVB}
we plot the graphs of $\LB(\lambda)$ (as a function of the number of iterations)
obtained from the exact natural gradient ascent algorithm (NGVB)
, inversion free gradient algorithm (IFVB) and weighted inversion free algorithm (AIFVB).
It can be seen that the three algorithms have
almost identical performance in this case. One striking note is that AIFVB even outperforms (obtaining a larger lower bound value)
the exact natural gradient algorithm in this case. It again helps
to confirm that the averaging technique is useful.  We note that we are tempted
to include the ELBO obtained from the Euclidean gradient ascent
but its performance is too unstable to include.
This phenomenon has been observed in \cite{tan2021analytic}
as the matrix $\Sigma$ obtained from
the traditional gradient ascent
is usually not positive definite,
which results in an unstable performance.

\begin{figure}[h]
	\centering
	\includegraphics[scale=0.4]{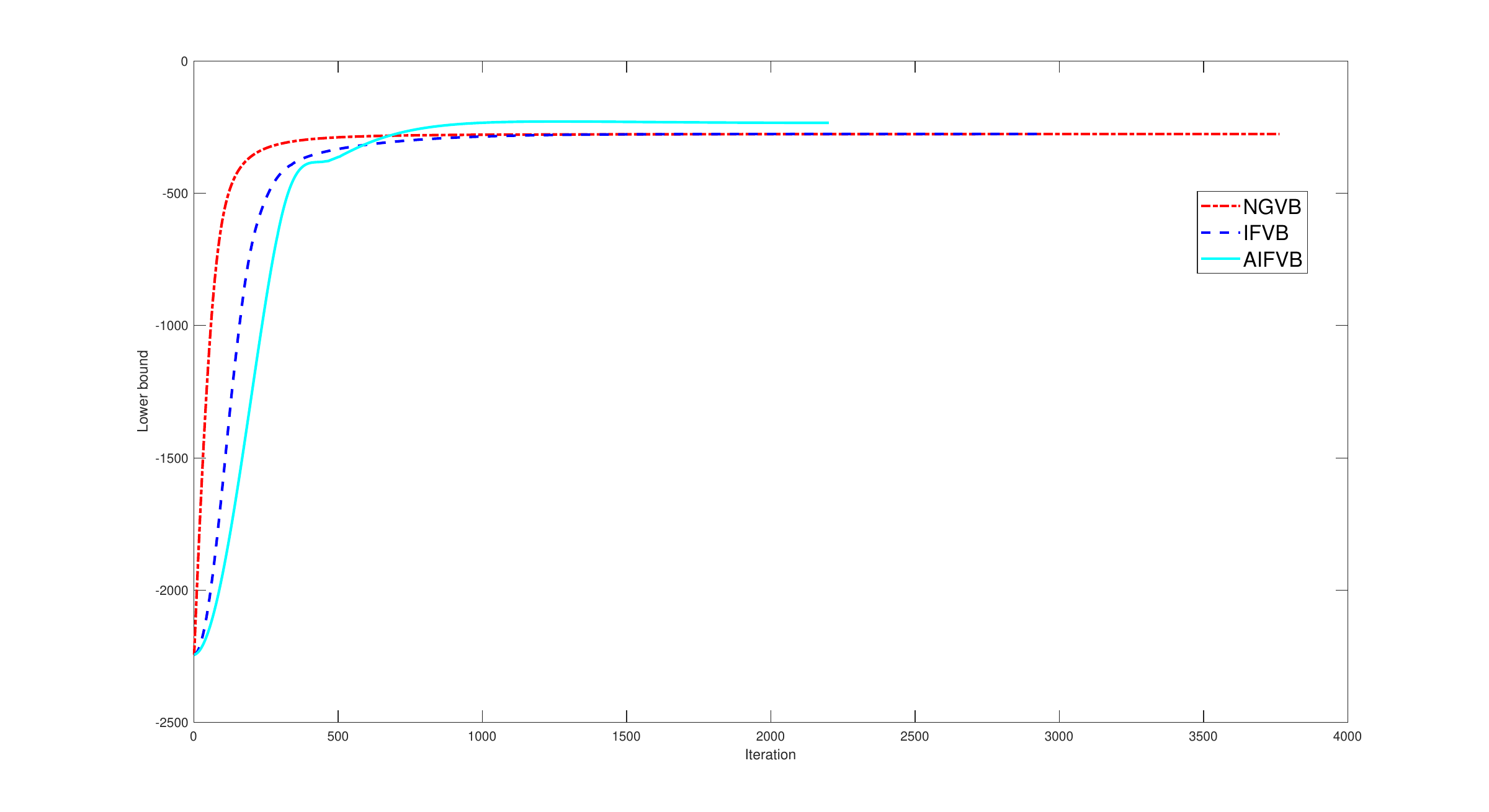}
\caption{Gaussian Approximation: Comparing NGVB, IFVB and AIFVB}\label{Gaussian Approximation: Comparing NGVB versus IFVB}
	
\end{figure}

\vspace{1cm}
\noindent\textbf{Example 3 (Normalizing flow Variational Bayes).} 
The choice of flexible variational distributions $q_\lambda$ is important in VB.
Normalizing flows \citep{rezende2015variational} are a class of techniques for designing flexible and expressive $q_\lambda$.
We consider a flexible VB framework where the VB approximation distribution $q_\lambda(\theta)$ is constructed based on a normalizing flow as follows
\begin{align}\label{eq: normal flows}
\epsilon\sim \mathcal N_d(0,I),\;\;Z=\psi(W_1\epsilon+b_1),\;\;\theta=W_2Z+b_2
\end{align}
where $W_1,W_2\in\text{Mat}(d,d)$ and $b_1,b_2$ are $d$-vectors,
and $\psi(\cdot)$ is an activation function such as sigmoid.
The transformation from $\epsilon$ to $\theta$ in \eqref{eq: normal flows} can be viewed as a neural network with one hidden layer $Z$; it might be desirable to consider deeper neural nets with more than one hidden layers, but we do not consider it here.
To be able to compute the density $q_\lambda(\theta)$ resulting from the transformations in \eqref{eq: normal flows},
these transformations should be invertibe and the determinants of the Jacobian matrices should be easy to compute.
To this end, we impose the orthogonality constraint on $W_1$ and $W_2$: $W_1^\top W_1=W_2^\top W_2=\mathbb{I}_d$, i.e. the columns are orthonormal.
Details on the derivation of the lower bound gradient and score function $\nabla\log q_\lambda(\theta)$ can be found in Appendix \ref{Appendix:Example 5}.

\subsubsection*{\bf Numerical results}
We apply the manifold normalizing flow VB (NLVB) \eqref{eq: normal flows} to approximate the posterior distribution in a neural network classification problem, using the German Credit dataset.
This dataset, available on the UCI Machine Learning Repository:
\begin{center}
\texttt{https://archive.ics.uci.edu/ml/index.php}, 
\end{center}

\noindent consists of observations on 1000 customers, each was already rated as being ``good credit'' (700 cases) or ``bad credit''
(300 cases). 
We create 10 predictors out of the available covariate variables including credit history, education, employment status, etc. 
The classification problem is based on a neural network with one hidden layer of 5 units. 
As $W_1$ and $W_2$ belong to the Stiefel manifold, for a comparison we use the VB on manifold algorithm of \cite{Tran:STCO2021} for updating these parameters.  
Figure \ref{fig:NNVB} plots the lower bounds of the IFVB algorithm (solid red) together with the conventional VB algorithm (dash blue) using the Euclidean gradient \eqref{eq:repram trick VB} in the Appendix.
As shown, the IFVB algorithm converges quicker and achieves a higher lower bound.  
Note that we do not consider the AIFVB algorithm in this example, as the variational parameters belong to the Stiefel manifold, making the averaging technique challenging to use. It is interesting to extend the work to handle cases where
the parameter space is a Riemannian manifold; however, we do not consider this in the present paper.

\begin{figure}[h]
	\centering
	\includegraphics[scale=0.5]{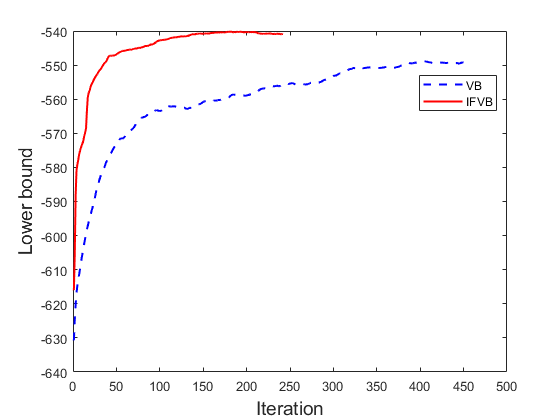}
\caption{Normalizing flow VB: Lower bound values of Euclidean gradient VB (dash blue) versus IFVB (solid red) over the iterations.}\label{fig:NNVB}
\end{figure}

\vspace{1.5cm}
\noindent\textbf{Example 4 (Bayesian neural network).}
This example considers a Bayesian neural network for regression
\begin{equation}\label{eq:BNN}
y = \eta(x,\omega)+\epsilon,\;\;\epsilon\sim\mathcal N(0,\sigma^2),
\end{equation}
where $y$ is a real-valued response variable, $\eta(x,\omega)$ denotes the output of a neural network with the input vector $x=(x_1,...,x_p)^\top\in\mathbb{R}^p$ and the vector of weights $w$.
As neural networks are prone to overfitting, we follow \cite{Tran:JCGS2020} and place a 
Bayesian adaptive group Lasso prior on the first-layer weights
\begin{equation}\label{eq:BNN prior}
w_{x_j}|\kappa_j\sim \mathcal N(0,\kappa_j\mathbb{I}_m),\;\;\;\;\kappa_j|\gamma_{j}\sim\text{Gamma}\left(\frac{m+1}{2},\frac{\gamma_{j}^2}{2}\right),\;\;j=1,...,p,
\end{equation}
with the $\gamma_j>0$ the shrinkage parameters; no regularization prior is put on the rest of the network weights. Here $w_{x_j}$ denotes the vector of weights that connect the input $x_j$ to the $m$ units in the first hidden layer. 
An inverse-Gamma prior is used for $\sigma^2$.
We use empirical Bayes for selecting the shrinkage parameters $\gamma_j$, and the posterior of $\kappa_j$ and $\sigma^2$ is approximated by a fixed-form within mean-filed VB procedure.
See \cite{Tran:JCGS2020} for the details.

The main task is to approximate the posterior of the network weights $\omega$. Let $d$ be the dimension of $\omega$.
We choose to approximate this posterior by a Gaussian variational distribution of the form $q_\lambda(\omega)=\mathcal N(\mu,\Sigma)$ with covariance matrix $\Sigma$ having a factor form $\Sigma=bb^\top+\diag(c)^2$,
where $b$ and $c$ are vectors in $\mathbb{R}^d$. The vector of the variational parameters is $\lambda=(\mu^\top,b^\top,c^\top)^\top$.
This factor structure of the covariance matrix significantly reduces the size of variational parameters, making the Gaussian variational approximation method computationally efficient for Bayesian inference in large models such as Bayesian neural networks.

\cite{Tran:JCGS2020} exploit the factor structure of $\Sigma$,
and by setting certain sub-blocks of the Fisher information matrix $I_F(\lambda)$ to zero,
to be able to derive a closed-form approximation of the inverse $I_F^{-1}(\lambda)$.
Their VB method, termed the NAtural gradient Gaussian Variational Approximation with
factor Covariance method (NAGVAC), is highly computationally efficient; however, the approximation of $I_F^{-1}(\lambda)$ might offset the VB approximation accuracy.

We now fit the Bayesian neural network model \eqref{eq:BNN}-\eqref{eq:BNN prior} to the Direct Marketing dataset \citep{Jank:2011} that consists of 1000 observations, of which 800 were used for training, and the rest for validation. The response $y$ is the amount (in \$1000) a
customer spends on the company's products per year, and 11 covariates include gender, income, married status, etc. 
We use a neural network with two hidden layers, each with ten units.
Figure \ref{NAGVAC versus IFVB} plots the mean squared error (MSE) values, computed on the validation set, of the VB training using ADAM \citep{kingma2014adam}, AIFVB, IFVB and NAGVAC.
With the same stopping rule, the
four methods, ADAM, AIFVB, IFVB and NAGVAC,
stop after 425, 221, 480 and 700 iterations, respectively.  
This confirms the theoretical result in Theorem \ref{theo::wasn} that the averaging technique speeds up the convergence of AIFVB conpared to IFVB.
On the validation set, the smallest MSE values produced by ADAM, AIFVB, IFVB and NAGVAC are 0.2273, 0.1750, 0.1749 and 0.1992, respectively.
%On the test set, the smallest MSE values produced by these four methods are 0.2989, 0.2979, 0.2678 and 0.2749.
Both IFVB and AIFVB perform better than ADAM and NAGVAC, probably because the natural gradient approximation in NAGVAC,
although being highly computationally efficient, might offset the approximation accuracy.

\begin{figure}[h]
	\centering
	\includegraphics[scale=0.5]{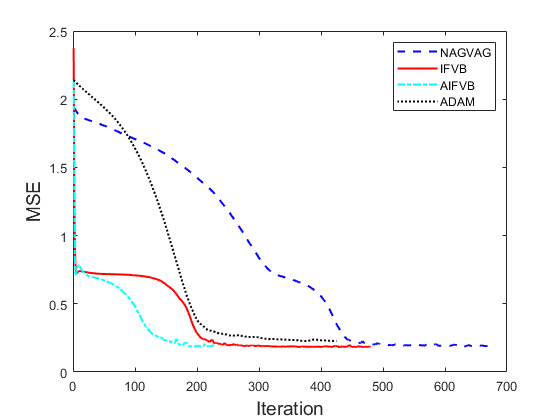}
\caption{Bayesian neural network: Validation MSE values of ADAM (dotted black), AIFVB (dot greenish), IFVB (red solid) and NAGVAC (dash blue) over the iterations.
}\label{NAGVAC versus IFVB}
\end{figure}

\section{Conclusion}\label{Conclusion}
The paper introduced an efficient approach for approximating the inverse of Fisher information, a crucial component in variational Bayes used for approximating posterior distributions. An outstanding feature of our algorithm is its avoidance of calculating the Fisher matrix and its inversion. Instead, our approach  generates a sequence of matrices converging to the inverse of Fisher information. 
Implementation of our method for natural gradient estimate does not require storage of large matrices.
Our inversion free VB framework showcases versatility, enabling its application in a wide range of domains, including Gaussian approximation and normalizing flow Variational Bayes,
and makes the natural gradient VB method applicable in cases that were impossible before. To demonstrate the efficiency and reliability of the method, we provided numerical examples as evidence of its effectiveness. We find it intriguing to consider expanding the scope of our approach to scenarios where the variational parameter space is a Riemannian manifold and to develop a rigorous theoretical framework for such cases. We plan to explore this avenue in our future research studies.

\bibliographystyle{plainnat}

\bibliography{Bibliography-MM-MC}
\newpage
\bigskip
\begin{center}
{\large\bf SUPPLEMENTARY MATERIAL}
\end{center}

\section{Example 5 (Gaussian and Inverse Gamma variational approximation)}\label{sec:example 5}
 Let $y = (11;12;8;10;9;8;9;10;13;7)$ be observations from
 $\mathcal N(\mu,\sigma^2)$,
 the normal distribution with mean $\mu$ and variance $\sigma^2$.
 We use
 the prior $\mathcal N(\mu_0,\sigma_0^2)$
 for $\mu$ and Inverse-Gamma$(\alpha_0,\beta_0)$
 for $\sigma^2$
 with the hyperparameters $\mu_0=0,\sigma_0=10,\alpha_0=1,\beta_0=1$.
 The posterior distribution 
 is written as 
 $$p(\mu,\sigma^2|y)\propto p(\mu)p(\sigma^2)p(y|\mu,\sigma^2).$$
 Assume that the VB approximation is
 $q_{\lambda}(\theta)=q(\mu)q(\sigma^2)$
 with $q(\mu)=\mathcal N(\mu_{\mu},\sigma_{\mu}^2)$,
 $q(\sigma^2)=\mbox{Inverse-Gamma}(\alpha_{\sigma^2},\beta_{\sigma^2})$, model parameter $\theta=(\mu,\sigma^2)^\top$
 and the variational parameters $\lambda=(\mu_\mu,\sigma^2_\mu,\alpha_{\sigma^2},\beta_{\sigma^2})$.
 Note that we have $h_\lambda(\theta)=h(\theta)-\log q_\lambda(\theta)$
 with
 \begin{align*}
  h(\theta)&=\log(p(\mu)p(\sigma^2)p(y|\mu,\sigma^2))\\
   &=-\frac{n+1}{2}\log(2\pi)-\frac{1}{2}\log(\sigma_0^2)
   -\frac{(\mu-\mu_0)^2}{2\sigma_0^2}
   +\alpha_0\log(\beta_0)-\log \Gamma(\alpha_0)\\
  & -(n/2+\alpha_0+1)\log(\sigma^2)
   -\frac{\beta_0}{\sigma^2}-\frac{1}{2\sigma^2}\sum_{i=1}^n(y_i-\mu)^2,
 \end{align*}
 and 
 \begin{align*}
 \log q_\lambda(\theta) &=
 \alpha_{\sigma^2}\log\beta_{\sigma^2}-\log \Gamma(\alpha_{\sigma^2})-(\alpha_{\sigma^2}+1)\log\sigma^2
 -\frac{\beta_{\sigma^2}}{\sigma^2}\\
 & -\frac{1}{2}\log(2\pi)
 -\frac{1}{2}\log(\sigma^2_{\mu})
 -\frac{(\mu-\mu_\mu)^2}{\sigma^2_\mu}.
 \end{align*}
 From here it can be seen that
 \begin{align*}
 \nabla_\lambda\log q_\lambda(\theta)=
 \left(\frac{\mu-\mu_\mu}{\sigma^2_\mu}
,-\frac{1}{2\sigma^2_\mu}+\frac{(\mu-\mu_\mu)^2}{2\sigma^4_\mu},
\log\beta_{\sigma^2}-\frac{\Gamma^\prime(\alpha_{\sigma^2})}{\Gamma(\alpha_{\sigma^2})}-\log\sigma^2,
\frac{\alpha_{\sigma^2}}{\beta_{\sigma^2}}-\frac{1}{\sigma^2} 
   \right)^\top.
 \end{align*}
 By direct calculation, it can be seen
 that the Fisher information matrix $I_{F}$
 is a diagonal block matrix with two main blocks
$$ \begin{pmatrix}
\frac{1}{\sigma^2_\mu}& 0\\
0 &\frac{1}{2\sigma^4_\mu}
\end{pmatrix},\;\;
\mbox{ and }\;\;
 \begin{pmatrix}
\frac{\partial^2\log \Gamma(\alpha_{\sigma^2})}{\partial(\alpha_{\sigma^2})^2}& -\frac{1}{\beta_{\sigma^2}}\\
-\frac{1}{\beta_{\sigma^2}} &\frac{\alpha_{\sigma^2}}{\beta^2_{\sigma^2}}
\end{pmatrix}.
$$
In this example, the natural gradient can be computed in closed-form,
which facilitates the testing of our IFVB and AIFVB algorithms.
We use $(\bar y,0.5,1,1)^\top$ as the initial guess for $\lambda$.
The first and second panel of Figure \ref{Fig:Comparing MCMC, NGVB and IFVB} plot the posterior
densities of $\mu$ and $\sigma^2$
using all four different approaches: MCMC,
exact natural gradient VB (NGVB), inversion free VB (IFVB) and averaged inversion free VB (AIFVB).
It can be seen that the posterior estimates obtained
by IFVB and AIFVB are close
to that of NGVB and MCMC.
The last panel of Figure \ref{Fig:Comparing MCMC, NGVB and IFVB}
plots the lower bound 
obtained from the VB methods. As shown, both IFVB and AIFVB converge  almost as fast as NGVB even though NGVB  uses the exact natural gradient
in its computation. 
%Note that in this example, AIFVB slightly underperforms as compared
%to that of IFVB algorithm.  In the other examples presented in this section, AIFVB outperforms IFVB.}
%We also tried to estimate the posterior distribution 
%using the approximated Fisher information $I_F$ in \eqref{FisherInformation-appproximation}.
%However, the result is too unstable to report.

\begin{figure}[h]
	\centering
	\includegraphics[scale=0.27]{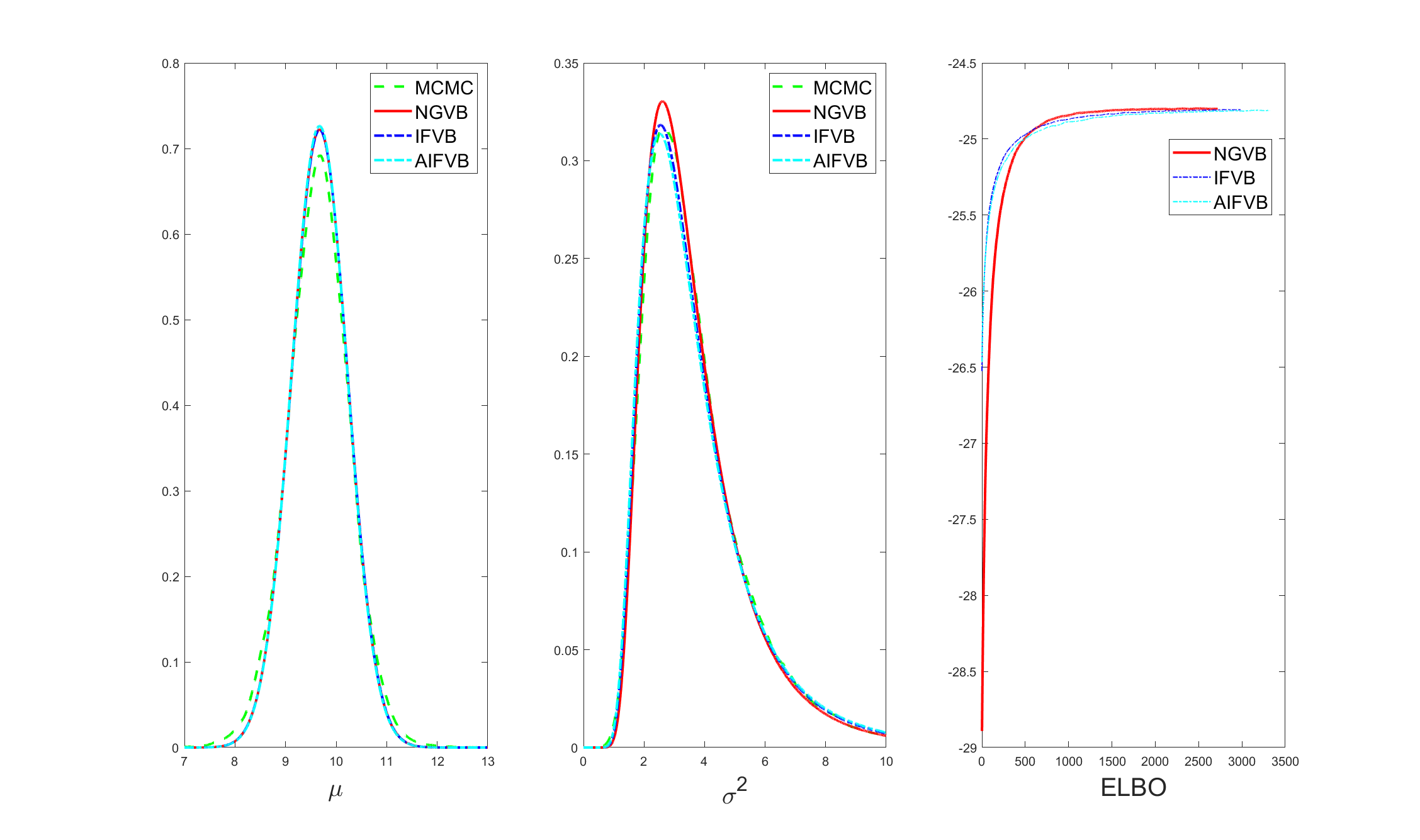}
\caption{The first two panels show that the IFVB and AIFVB estimates are close to that of the
exact natural gradient ascent VB (NGVB) and MCMC. The last panel plots the lower bound estimates obtained from NGVB, IFVB and AIFVB.}\label{Fig:Comparing MCMC, NGVB and IFVB}	
\end{figure}

\section{Main proofs}\label{sec:Main proofs}
Let us recall that $\tau_{k} = \frac{c_{\alpha}}{\left(c_{\alpha}' + k \right)^{\alpha}}$ with $c_{\alpha}> 0$, $c_{\alpha}' \geq 0$ and $\alpha \in (1/2,1)$. In addition, $c_{\beta} \geq 0$ and $\beta \in (0,\alpha -1/2 )$.
Recall that
\begin{align*}
\nabla_\lambda \mathcal L(\lambda)=\mathbb E_{q_\lambda}[ - \nabla_\lambda\log q_\lambda(\theta) h_\lambda(\theta)] =: \mathbb E_{q_\lambda}[\ell (\theta,\lambda)].
\end{align*}
We then have the recursive scheme
\begin{align*}
\lambda^{(k+1)}-\lambda^{(k)}=-\tau_{k+1}\tilde{\mathbf{H}}^{-1}_{k+1} \frac{1}{B}\sum_{i=1}^{B}\ell(\theta_{k+1,i},\lambda^{(k)}),\;\;\;\theta_{k+1,i}\sim q_{\lambda^{(k)}}(\theta).
\end{align*}

\subsection{Proof of Theorem \ref{theo::as}}
\begin{proof}
From Taylor's expansion of $\mathcal L(\lambda)$, since its Hessian is uniformly bounded, we have
\begin{align*}
&\mathcal L(\lambda^{(k+1)})
=\mathcal L(\lambda^{(k)})+\nabla_\lambda\mathcal L^\top(\lambda^{(k)})(\lambda^{(k+1)}-\lambda^{(k)})\\
&+ (\lambda^{(k+1)}-\lambda^{(k)})^\top\Big(\int_0^1 (1-t)\nabla_\lambda^2\mathcal L(\lambda^{(k+1)}+t(\lambda^{(k)}-\lambda^{(k+1)}))dt\Big) (\lambda^{(k+1)}-\lambda^{(k)})\\
&\leq \mathcal L(\lambda^{(k)})+\nabla_\lambda\mathcal L^\top(\lambda^{(k)})(\lambda^{(k+1)}-\lambda^{(k)})+\frac{1}{2}L_0\norm{(\lambda^{(k+1)}-\lambda^{(k)})}^2\\
&\leq  \mathcal L(\lambda^{(k)})-\alpha_{k+1}\nabla_\lambda\mathcal L^\top(\lambda^{(k)})\tilde{\mathbf{H}}^{-1}_{k+1}\frac{1}{B}\sum_{i=1}^{B}\ell(\theta_{k+1,i},\lambda^{(k)})\\
&\phantom{cccc}+
\frac{1}{2B}\alpha_{k+1}^2 L_0\norm{\tilde{\mathbf{H}}^{-1}_{k+1}}_{op}^2 \sum_{i=1}^{B}\norm{\ell(\theta_{k+1,i},\lambda^{(k)})}^2
\end{align*}
Hence
\begin{align*}
\mathcal L(\lambda^{(k+1)})-\mathcal L(\lambda^*)
& \leq  \mathcal L(\lambda^{(k)})-\mathcal L(\lambda^*)-\alpha_{k+1}\nabla_\lambda\mathcal L^\top(\lambda^{(k)})\tilde{\mathbf{H}}^{-1}_{k+1} \frac{1}{B}\sum_{i=1}^{B}\ell(\theta_{k+1,i},\lambda^{(k)}) \\
& + \frac{1}{2B}\alpha_{k+1}^2 L_0\norm{\tilde{\mathbf{H}}^{-1}_{k+1}}_{op}^2 \sum_{i=1}^{B}\norm{\ell(\theta_{k+1,i},\lambda^{(k)})}^2.
\end{align*}
Let us denote $W_k=\mathcal L(\lambda^{(k)})-\mathcal L(\lambda^*)$. From the above inequality we have
\begin{align*}
W_{k+1}\leq W_k-\alpha_{k+1}\nabla_\lambda\mathcal L^\top(\lambda^{(k)})\tilde{\mathbf{H}}^{-1}_{k+1} \frac{1}{B}\sum_{i=1}^{B}\ell(\theta_{k+1,i},\lambda^{(k)})\\
+\frac{1}{2B}\alpha_{k+1}^2 L_0\norm{\tilde{\mathbf{H}}^{-1}_{k+1}}_{op}^2 \sum_{i=1}^{B}\norm{\ell(\theta_{k+1,i},\lambda^{(k)})}^2.
\end{align*}
Let us consider the filtration $\mathcal F_k=\sigma(\theta_{s,i},\overline{\theta}_{s'},Z_{s}:s\leq k:s' \leq k+1)$.
As $\theta_{k+1,i}$ are independent of $\mathcal F_k$, taking the conditional expectation from the both sides
of the above inequality, we have
\begin{align*}
\mathbb E[W_{k+1}|\mathcal F_k]
&\leq \mathbb E[W_k|\mathcal F_k]
-\alpha_{k+1} \frac{1}{B}\sum_{i=1}^{B}\mathbb E[\nabla_\lambda\mathcal L^\top(\lambda^{(k)})\tilde{\mathbf{H}}^{-1}_{k+1}\ell(\theta_{k+1,i},\lambda^{(k)})|\mathcal F_k]\\
&\phantom{cccc}+\frac{1}{2B}\alpha_{k+1}^2 L_0 \sum_{i=1}^{B}\norm{\tilde{\mathbf{H}}^{-1}_{k+1}}_{op}^2\mathbb E[\norm{\ell(\theta_{k+1,i},\lambda^{(k)})}^2|\mathcal F_k]\\
&\leq W_k
- \alpha_{k+1}\nabla_\lambda\mathcal L^\top(\lambda^{(k)})\tilde{\mathbf{H}}^{-1}_{k+1}\nabla_\lambda\mathcal L(\lambda^{(k)})\\
&\phantom{cccc}+
\frac{1}{2B}\alpha_{k+1}^2 L_0\norm{\tilde{\mathbf{H}}^{-1}_{k+1}}_{op}^2 \sum_{i=1}^{B}\mathbb E\left[\norm{\ell(\theta_{k+1,i},\lambda^{(k)})}^2|\mathcal F_k \right].
\end{align*}
Since
$\mathbb E[\norm{\ell(\theta,\lambda)}^2]\leq C_0+C_1(\mathcal L(\lambda)-\mathcal L(\lambda^*))$, we have
\begin{align*}
\mathbb E[W_{k+1}|\mathcal F_k]
&\leq \left(1+\frac{ C_1L_0}{2}\alpha_{k+1}^2\norm{\tilde{\mathbf{H}}^{-1}_{k+1}}_{op}^2 \right) W_k
-\alpha_{k+1}\lambda_{\min}(\tilde{\mathbf{H}}^{-1}_{k+1})\norm{\nabla_\lambda\mathcal L(\lambda^{(k)})}^2\\
&+\frac{1}{2 }\alpha_{k+1}^2 C_0L_0\norm{\tilde{\mathbf{H}}^{-1}_{k+1}}_{op}^2.
\end{align*}
In order to apply the Robbins-Siegmund theorem \citep{robbins1971convergence}, we now focus on the behavior of the eigenvalues of $\tilde{\mathbf{H}}_{s}$.   Observe that with the help of Toeplitz lemma (e.g., see the proof of Theorem \ref{thm:H1} in Appendix \ref{Proof of Riccati-twice-algorithm}).
\[
(1-\beta ) s^{1-\beta} \sum_{k=1}^{s}k^{-\beta} Z_{k}Z_{k}^{\top} \xrightarrow[s\to + \infty]{} I.
\]
Then, as soon as $\lambda_{\min} \left( \tilde{\mathbf{H}}_{s} \right) \geq \lambda_{\min} \left( \frac{1}{s}H_0 + \frac{1}{s}\sum_{k=1}^{s} c_{\beta}k^{-\beta} Z_{k}Z_{k}^{\top} \right)$, it comes
\[
\lambda_{\max} \left( \tilde{\mathbf{H}}_{s}^{-1} \right) = O \left( s^{\beta} \right) \quad a.s.
\]
Then, since $\beta < \alpha -1/2$, i.e $2\alpha -2\beta > 1$, it comes
\[
\sum_{k} \alpha_{k+1}^{2} \left\|   \tilde{\mathbf{H}}_{k+1}^{-1} \right\|^{2} \leq C\sum_{k} k^{2\beta-2\alpha} <+\infty \quad a.s.
\]
Using the Robbins-Siegmund theorem, one has that $W_{k}$ converges almost surely to a finite random variable $W_{\infty}$, which proves (i). We also conclude that
\begin{equation}\label{eq::penh}
\sum_{k=0}^\infty \alpha_{k+1}\lambda_{\min}(\tilde{\mathbf{H}}^{-1}_{k+1}) \left\| \nabla \mathcal{L}\left( \lambda^{(k)} \right) \right\|^{2}  <\infty,\quad \mbox{a.s}.
\end{equation}
We now have to control the smallest eigenvalue of $\tilde{\mathbf{H}}^{-1}_{s+1}$. In this aim, let us remark that the estimates of the Fisher Information can be written as
%\MNT{the index $s$ in $\tilde{\mathbf{H}}_{s}$ below needs to be double checked to be consistent with the indexing in Sec 4.SHould be s+1 in my mind}
\begin{align*}
\tilde{\mathbf{H}}_{s}= \frac{1}{s}\left( H_0 + \sum_{k=1}^{s} \beta_{k}Z_{k}Z_{k}^{T} \right)& + \underbrace{\frac{1}{s}\sum_{k=0}^{s-1} \mathbb{E}\left[ \nabla_\lambda\log q_{\overline{\lambda}^{(k)}}(\overline{\theta}_{k+1})\nabla_\lambda\log q_{\overline{\lambda}^{(k)}}(\overline{\theta}_{k+1})^{\top} |\tilde{\mathcal{F}}_{k} \right]}_{=: R_{s}} \\
&+ \underbrace{\frac{1}{s} \sum_{k=0}^{s-1} \Xi_{k+1}}_{=: M_{s}} ,
\end{align*}
where $\tilde{\mathcal{F}}_{k} = \sigma \left( \theta_{s,i},\overline{\theta}_{s},Z_{s}:s\leq k,i \leq B \right)$ 
%\MNT{$\theta_s$ in the filtration should be $\overline{\theta}_s$? Yes I corrected it}
and   
\[\Xi_{k+1} := 
 - \mathbb{E}\left[  \nabla_\lambda\log q_{\overline{\lambda}^{(k)}}(\overline{\theta}_{k+1})\nabla_\lambda\log q_{\overline{\lambda}^{(k)}}(\overline{\theta}_{k+1})^{\top} | \tilde{\mathcal{F}}_{k} \right] +  \nabla_\lambda\log q_{\overline{\lambda}^{(k)}}(\overline{\theta}_{k+1})\nabla_\lambda\log q_{\overline{\lambda}^{(k)}}(\overline{\theta}_{k+1})^{\top}\] 
 is a sequence of martingale differences. In addition
\begin{align}\label{eq: M matingale eq}
\mathbb{E}\left[ \left\| M_{s+1} \right\|_{F}^{2}  |\mathcal{F}_{s} \right] & \leq \left( \frac{s}{s+1} \right)^{2} \left\| M_{s} \right\|_{F}^{2}  
  + \frac{1}{(s+1)^{2}} \mathbb{E}\left[  \left\| \nabla_\lambda\log q_{\overline{\lambda}^{(s)}}(\overline{\theta}_{s+1}) \right\|^{4} | \tilde{\mathcal{F}}_{s} \right]  .
\end{align}
Since there are positive constants $C_{0}',C_{1}'$ such that for all $\lambda$, 
\[
\mathbb{E}\left[ \left\| \nabla_{\lambda} \log q_{\lambda} (\theta) \right\|^{4} \right] \leq C_{0}' + C_{1}' \left( \mathcal{L}(\lambda) - \mathcal{L}\left( \lambda^{*} \right)\right)^{2} ,
\]
we have that $\mathbb{E}\left[ \left\|  \nabla_\lambda\log q_{\overline{\lambda}^{(k)}}(\overline{\theta}_{k+1}) \right\|^{4} | \tilde{\mathcal{F}}_{k} \right] \leq C_{0}' + C_{1}' \overline{W}_{k}^{2}  $ where $\overline{W}_{k} = \mathcal{L} \left( \overline{\lambda}^{(k)} \right) - \mathcal{L}\left( \lambda^{*} \right)  $.  Then, if $\mathcal{L}$ is not convex, one has $C_{1}' = 0$, otherwise, by the convexity of $\mathcal{L}$ and with the help of the Toeplitz lemma,
\[
\overline{W}_{s} \leq \frac{1}{\sum_{k=0}^{s} \log (k+1)^{w}} \sum_{k=0}^{s}\log (k+1)^{w} \left(  \mathcal{L}\left( \lambda^{(k)} \right) - \mathcal{L}\left( \lambda^{*} \right) \right) \xrightarrow[n\to + \infty]{a.s} W_{\infty}.
\]
Then, from \eqref{eq: M matingale eq}, applying the Robbins-Siegmund theorem, $\left\| M_{s} \right\|_{F}^{2} = \mathcal{O}(1)$ a.s. In a same way, one can check that $\left\| R_{s} \right\|_{F}^{2} = \mathcal O(1)$ a.s. This means that at least, 
\[
\liminf_{k}  \lambda_{\min} \left( \tilde{\mathbf{H}}_{k}^{-1} \right)    > 0
\]
%\MNT{Is there the factor $\log k$ here? No.}
so \eqref{eq::penh} implies that
\[
\sum_{k =0}^{+\infty} \alpha_{k+1} \left\| \nabla \mathcal{L} \left( \lambda^{(k)} \right) \right\|^{2} < + \infty \quad a.s.
\]
Then, thanks to Lemma 2 in \cite{liu2022almost} (or following the proof of Theorem 4.3 in \cite{sebbouh2021almost}), it comes that
\[
\min_{k=0,\ldots s} \left\| \nabla \mathcal{L} \left( \lambda^{(k)} \right) \right\|^{2} = o \left( \frac{1}{\sum_{k=0}^{s} \alpha_{k+1}} \right) \quad a.s.
\]
i.e. one has
\[
\min_{k=0,\ldots s} \left\| \nabla \mathcal{L} \left( \lambda^{(k)} \right) \right\|^{2} = o \left( \frac{1}{s^{1-\alpha}} \right) \quad a.s.
\]
which proves (ii). One can apply Lemma 2 in \cite{liu2022almost} again to obtain the result on $\overline{\lambda}^{(k)}$.
This completes the proof of the theorem.

%and that $\lambda^{(k)}$ converges almost surely to $\lambda^{*}$. 

\end{proof}

\subsection{Proof of Corollary \ref{cor::as::hs}}
\begin{proof}
Let us recall from the proof of Theorem \ref{theo::as} that $\tilde{\mathbf{H}}_{s} = \frac{1}{s}\left( H_0 + c_{\beta}\sum_{k=1}^{s} k^{-\beta}Z_{k}Z_{k}^{T} \right) + R_{s} + M_{s}$.
Similar to the proof of Theorem \ref{thm:H1} in Appendix \ref{Proof of Riccati-twice-algorithm}, the norm of the first term
can be estimated as
\[
 \left\| \frac{1}{s}\left( H_0 + c_{\beta}\sum_{k=1}^{s} k^{-\beta}Z_{k}Z_{k}^{T} \right) \right\| = \mathcal{O} \left( \max \left\lbrace c_{\beta}s^{-\beta}, s^{-1} \right\rbrace \right) \quad a.s.
\]
which is is negligible. By the continuity of $I_{F}(\lambda)$ and since the convergence of $\lambda^{(k)}$ implies that $\overline{\lambda}^{(k)}$ converges almost surely to $\lambda^{*}$, it comes
\[
%R_{s} = \frac{1}{s+1}\sum_{k=0}^{s-1} I_{F} \left( \overline{\lambda}^{(k)} \right) \xrightarrow[n\to + \infty]{a.s} I_{F} \left( \lambda^{*} \right) .
R_s = \frac{1}{s}\sum_{k=0}^{s} \mathbb{E}\left[ \nabla_\lambda\log q_{\overline{\lambda}^{(k)}}(\overline{\theta}_{k+1})\nabla_\lambda\log q_{\overline{\lambda}^{(k)}}(\overline{\theta}_{k+1})^{\top} |\tilde{\mathcal{F}}_{k} \right]\xrightarrow[n\to + \infty]{a.s} I_{F} \left( \lambda^{*} \right).
\]
Finally, applying Theorem 6.2 in \cite{cenac2020efficient}, it comes that for any $\delta > 0$,
\[
\left\| M_{s} \right\|_{F}^{2} = o \left( \frac{\log s^{1+\delta}}{s} \right) \quad a.s.
\]
which is negligible.
\end{proof}

\subsection{Proof of Theorem \ref{theo::rate::lambda}}
\begin{proof} 
The proof is adapted from \cite{boyer2023asymptotic} and \cite{bercu2020efficient}. Denoting $I_{F} = I_{F}\left( \lambda^{*} \right)$,
\begin{align}\label{dec::delta::areutiliser}
\notag &\lambda^{(k+1)}-\lambda^*=\lambda^{(k)} - \lambda^{*}-\tau_{k+1}\tilde{\mathbf{H}}_{k+1}^{-1}\frac{1}{B}\sum_{i=1}^{B}\ell(\theta_{k+1,i},\lambda^{(k)}) \\
\notag &=\lambda^{(k)}-\lambda^* -\tau_{k+1}\tilde{\mathbf{H}}_{k+1}^{-1}\nabla_\lambda\mathcal L(\lambda^{(k)})+\tau_{k+1}\tilde{\mathbf{H}}_{k+1}^{-1}\underbrace{\left(\nabla_\lambda\mathcal L(\lambda^{(k)})- \frac{1}{B}\sum_{i=1}^{B}\ell(\theta_{k+1,i},\lambda^{(k)}) \right)}_{\xi_{k+1}}\\
\notag &=\lambda^{(k)}-\lambda^* - \tau_{k+1} I^{-1}_F \nabla_\lambda\mathcal L(\lambda^{(k)})-\tau_{k+1} \underbrace{\left(\tilde{\mathbf{H}}_{k+1}^{-1}-I^{-1}_F \right)\nabla_\lambda\mathcal L(\lambda^{(k)})}_{=: r_{k}}+\tau_{k+1}\tilde{\mathbf{H}}_{k+1}^{-1}\xi_{k+1}\\
\notag & = \left( \mathbb I - \tau_{k+1}I^{-1}_{F}  \nabla^{2}\mathcal{L}\left( \lambda^{*} \right) \right) \left( \lambda^{(k)} - \lambda^{*} \right) - \tau_{k+1}r_{k} + \tau_{k+1}\tilde{\mathbf{H}}_{k+1}^{-1}\xi_{k+1}  \\
&\phantom{cccc} - \tau_{k+1} I_{F}^{-1}  \underbrace{  \left( \nabla \mathcal{L}\left( \lambda^{(k)} \right) - \nabla^{2}\mathcal{L}\left( \lambda^{*} \right) \left( \lambda^{(k)} - \lambda^{*} \right) \right)}_{=: \delta_{k}}
\end{align}
As explained in \cite{antonakopoulos2022adagrad}, since $I_{F}$ and $\nabla^{2} \mathcal{L}\left( \lambda^{*} \right)$ are symmetric and positive, $ I^{-1}_{F}  \nabla^{2}\mathcal{L}\left( \lambda^{*} \right)$ and $I^{-1/2}_{F}  \nabla^{2}\mathcal{L}\left( \lambda^{*} \right)I^{-1/2}_{F}  $ have the same eigenvalues, i.e there is a positive diagonal matrix $\Lambda$ (of the eigenvalues of $I^{-1/2}_{F}  \nabla^{2}\mathcal{L}\left( \lambda^{*} \right) I^{-1/2}_{F}  $) and a matrix $Q$ such that $I^{-1}_{F}  \nabla^{2}\mathcal{L}\left( \lambda^{*} \right) = Q^{-1}\Lambda Q$.
Then one can rewrite the previous decomposition as 
\begin{equation}\label{decdelta}
Q \left( \lambda^{(k+1)}-\lambda^* \right)= \left( I - \tau_{k+1}\Lambda \right) Q \left( \lambda^{(k)} - \lambda^{*} \right) - \tau_{k+1}Q r_{k} + \tau_{k+1}Q \tilde{\mathbf{H}}_{k+1}^{-1}\xi_{k+1}   - \tau_{k+1}Q I_{F}^{-1} \delta_{k} .
\end{equation}
Then, with the help of an induction, it comes that
\begin{align}\label{decbeta}
Q \left(\lambda^{(s)} - \lambda^{*} \right) &= \beta_{s,0} Q \left( \lambda^{(0)} - \lambda^{*} \right) + \underbrace{\sum_{k=0}^{s-1} \beta_{s,k+1}\tau_{k+1}Q \tilde{\mathbf{H}}_{k+1}^{-1}\xi_{k+1}}_{=: M_{s}'}\notag \\
&- \underbrace{\left( \sum_{k=0}^{s-1} \beta_{s,k+1}\tau_{k+1}Q r_{k} + \sum_{k=0}^{s-1} \beta_{s,k+1}\tau_{k+1}QI_{F}^{-1} \delta_{k} \right)}_{\Delta_{s}}.
\end{align}
with $\beta_{s,k} = \prod_{j=k+1}^{s} \left(\mathbb I - \alpha_{j}\Lambda \right)$ and $\beta_{s,s} = \mathbb I$.  We now give the rate of convergence for each term in the decomposition \eqref{decbeta}.

\medskip

\noindent{\textbf{Rate of convergence of $\beta_{s,0}Q \left( \lambda^{(0)} - \lambda^{*} \right)$.}}
Since $\lambda_{\min}(\Lambda) > 0$ (because $I_{F}$ and $\nabla^{2}\mathcal{L}\left( \lambda^{*} \right)$ are positive),  one can easily check that 
\[
\left\| \beta_{s,0}Q \left( \lambda^{(0)} - \lambda^{*} \right) \right\| = \mathcal{O} \left( \exp \left( - \lambda_{\min}(\Lambda) \frac{c_{\alpha}}{1-\alpha}s^{1-\alpha} \right) \right) \quad a.s.
\]

\medskip

\noindent\textbf{Rate of convergence of $M_{s}'$.}
Recall there  are $\eta > \frac{1}{\alpha}-1$ and positive constants $C_{\eta,0},C_{\eta,1}$ such that for all $\lambda$,
\[
\mathbb{E}\left[ \left\| l(\theta , \lambda) \right\|^{2+2\eta} \right] \leq C_{\eta,0} + C_{\eta,1}\left( \mathcal{L} (\lambda) - \mathcal{L}\left( \lambda^{*} \right) \right)^{1+\eta}.
\]
Then, one has with the help of Holder's inequality
\begin{align*}
\mathbb{E}\left[ \left\| \xi_{k+1} \right\|^{2+2\eta} |\mathcal{F}_{k} \right]  & \leq 2^{2\eta +1} \frac{1}{B}\sum_{i=1}^{B}  \mathbb{E}\left[ \left\| \ell \left( \theta_{k+1,i},\lambda^{(k)} \right)\right\|^{2+2\eta} |\mathcal{F}_{k} \right] + 2^{2\eta+1} \left\| \nabla_{\lambda} \mathcal{L}\left( \lambda^{(k)} \right) \right\|^{2+2\eta} \\
& \leq 2^{2+2\eta}C_{\eta,0} + 2^{2+2\eta}C_{\eta,1}\left( \mathcal{L} \left(\overline{\lambda}^{(k)}\right) - \mathcal{L}\left( \lambda^{*} \right) \right)^{1+\eta} .
\end{align*}
Since   $\lambda^{(k)}$ is strongly consistent, the second term on the right-hand side of previous inequality converges almost surely to $0$. In addition, thanks to Corollary \ref{cor::as::hs},  $\tilde{\mathbf{H}}_{k}^{-1}$ converges almost surely to $I_{F}^{-1}$. Then, with the help of Theorem 6.1 in \cite{cenac2020efficient}, 
\[
\left\| M_{s}' \right\|^{2} = \mathcal{O} \left( \frac{\log s}{s^{\alpha}} \right) \quad a.s.
\]

\medskip

\noindent\textbf{Rate of convergence of $\Delta_{s}$. }
For $s$ large enough, one has
\begin{align*}
\left\| \Delta_{s+1} \right\|  \leq \left( 1- \lambda_{\min}(\Lambda) \tau_{s+1} \right) \left\| \Delta_{s} \right\| + \tau_{s+1} \left\| Q \right\| \left( \left\| r_{s} \right\| + \left\| I_{F}^{-1} \right\| \left\| \delta_{s} \right\| \right).
\end{align*}
Observe that since $\mathcal{L}$ is twice continuously differentiable on a neighborhood of $\lambda^{*}$ and since $\lambda^{(k)}$ is strongly consistent, it comes that $\left\| \delta_{k} \right\| = o \left( \left\| \lambda^{(k)} - \lambda^{*} \right\| \right)$ a.s. In addition, since $\tilde{\mathbf{H}}_{k}^{-1}$ converges almost surely to $I_{F}^{-1}$ and since   the gradient of $\mathcal{L}$ is locally Lipschitz on a neighborhood of $\lambda^{*}$ (since the Hessian is locally bounded by continuity), one has that $\| r_{k} \| = o \left( \left\| \lambda^{(s)} - \lambda^{*} \right\|\right) $ a.s. Then,   $ \left\| r_{s} \right\| + \left\| I_{F}^{-1} \right\| \left\| \delta_{s} \right\| = o \left( \left\| \lambda^{(s)} - \lambda^{*} \right\| \right)$ a.s, and with the help of decompotision \eqref{decbeta}, it comes that
\[
\left\| r_{s} \right\| + \left\| I_{F}^{-1} \right\| \left\| \delta_{s} \right\| = o \left(  \left\| \beta_{s,0} Q\left(  \lambda^{(0)} - \lambda^{*} \right) \right\| + \left\| M_{s}' \right\| + \left\| \Delta_{s} \right\| \right) \quad a.s.
\]
Thanks to previous convergence results, there exists a sequence of random variables $r_{s}'$ converging almost surely to $0$ such that
\[
\left\| \Delta_{s+1} \right\|  \leq \left( 1- \lambda_{\min}(\Lambda) \tau_{s+1} \right) \left\| \Delta_{s} \right\| + \tau_{s+1} r_{s+1} \left( \sqrt{\frac{\log s}{s^{\alpha}}} + \left\| \Delta_{s} \right\| \right)  .
\]
Then, thanks to a stabilization Lemma (see \cite{duflo2013random}), it comes
\[
\left\| \Delta_s \right\| = \mathcal{O} \left( \sqrt{\frac{\log s}{s^{\alpha}}} \right) \quad a.s,
\]
which concludes the proof.
\end{proof}

\subsection{Proof of Theorem \ref{theo::wasn}}
\begin{proof}
Observe that one can rewrite decomposition \eqref{dec::delta::areutiliser}  as
\begin{align*}
\lambda^{(k+1)} - \lambda^{*} &  = \lambda^{(k)} - \lambda^{*} - \tau_{k+1}\tilde{\mathbf{H}}_{k+1}^{-1} \nabla_{\lambda} \mathcal{L} \left( \lambda^{(k)} \right) + \tau_{k+1} \tilde{\mathbf{H}}_{k+1}^{-1} \xi_{k+1} \\
& = \lambda^{(k)} - \lambda^{*} - \tau_{k+1}\tilde{\mathbf{H}}_{k+1}^{-1} \nabla_{\lambda}^{2}\mathcal{L}(\lambda^{*}) \left( \lambda^{(k)} - \lambda^{*} \right) - \tau_{k+1}\tilde{\mathbf{H}}_{k+1}^{-1} \delta_{k} +   \tau_{k+1} \tilde{\mathbf{H}}_{k+1}^{-1} \xi_{k+1},
\end{align*}
where $\delta_k=\nabla \mathcal{L}\left( \lambda^{(k)} \right) - \nabla^{2}\mathcal{L}\left( \lambda^{*} \right) \left( \lambda^{(k)} - \lambda^{*} \right)$ and 
$\xi_{k+1}=\nabla_\lambda\mathcal L(\lambda^{(k)})- \frac{1}{B}\sum_{i=1}^{B}\ell(\theta_{k+1,i},\lambda^{(k)})$.
Denoting $u_{k}: = \lambda^{(k)} - \lambda^{*}$ and $L^{-1} := \nabla_{\lambda}^{2} \mathcal{L} (\lambda^{*})^{-1}$, one can rewrite  
\[
\lambda^{(k)} - \lambda^{*} = L^{-1}\tilde{\mathbf{H}}_{k+1}^{-1}  \frac{u_{k}- u_{k+1}}{\tau_{k+1}} + L^{-1} \xi_{k+1} - L^{-1} \delta_{k}.
\]
Multiplying by $\log (k+1)^{w}$, then summing these equalities and dividing by $\sum_{k=0}^{s}\log (k+1)^w$, it comes
  \begin{align*}\label{dec::moy}
\notag  \overline{\lambda}^{(s)}    - \lambda^{*} &   =  L^{-1} \underbrace{\frac{1}{\sum_{k=0}^{s} \log (k+1)^{w}} \sum_{k=0}^{s}\log (k+1)^{w} \tilde{\mathbf{H}}_{k+1}^{-1} \frac{u_{k} - u_{k+1}}{\tau_{k+1}}}_{=: A_{1,s}}\\  & -  L^{-1} \underbrace{\frac{1 }{\sum_{k=0}^{s} \log (k+1)^{w}} \sum_{k=0}^{s}\log (k+1)^{w} \delta_{k}}_{A_{2,s}}
          +     \underbrace{\frac{ L^{-1} }{\sum_{k=0}^{s} \log (k+1)^{w}} \sum_{k=0}^{s}\log (k+1)^{w} \xi_{k+1}}_{M_{2,s}}.
 \end{align*}
  The aim is to give the rate of convergence for each term on the right-hand side of the decomposition above. Let us first denote $t_{s} =\sum_{k=0}^{s} \log (k+1)^{w}$ and observe that 
  \begin{equation}\label{equiv::ts}
    \  t_{s} \sim s \log(s+1)^w .
  \end{equation}
  \medskip

  \noindent\textbf{Rate of convergence of $\mathbf{A_{1,s}}$. } First, note that
  \[
  A_{1,s} =  \frac{1}{t_{s}} \sum_{k=0}^{s}\log(k+1)^w   \frac{\tilde{\mathbf{H}}_{k}^{-1} u_{k} - \tilde{\mathbf{H}}_{k+1}^{-1}u_{k+1}   }{\tau_{k+1}} + \frac{1}{t_{s}}\sum_{k=0}^{s} \log(k+1)^w \left( \tilde{\mathbf{H}}_{k+1}^{-1} - \tilde{\mathbf{H}}_{k}^{-1} \right) \frac{u_{k}}{\tau_{k+1}} .
  \]
  Concerning the first term on the right hand-side of previous equality, with the help of Abel's transform,
  \begin{align*}
      \frac{1}{t_{s}} \sum_{k=0}^{s}\log(k+1)^w   \frac{\tilde{\mathbf{H}}_{k}^{-1} u_{k} - \tilde{\mathbf{H}}_{k+1}^{-1}u_{k+1}   }{\tau_{k+1}}  &  = - \frac{ \log(s+1)^w \tilde{\mathbf{H}}_{s+1}^{-1} u_{s+1}}{t_{s} \tau_{s+1}} + \frac{\mathbf{1}_{w=0}\tilde{\mathbf{H}}_{0}^{-1}u_{0}}{t_{s}\alpha_{1}} \\
      & + \sum_{k=1}^{s}  \tilde{\mathbf{H}}_{k}^{-1}u_{k} \left( \frac{\log(k+1)^{w}}{\tau_{k+1}} - \frac{\log(k)^w}{\tau_{k}} \right) .
  \end{align*}
%  \MNT{can the term $\mathbf{1}_{w=0}$ be removed? No but it does not really matter this term is totally negligible.}
With the help of an Abel's transform, one has
\begin{align}
    \frac{1}{t_{s}} \sum_{k=0}^{s}\log(k+1)^w   \frac{\tilde{\mathbf{H}}_{k}^{-1} u_{k} - \tilde{\mathbf{H}}_{k+1}^{-1}u_{k+1}   }{\tau_{k+1}} &  = \frac{- \log(s+1)^w\tilde{\mathbf{H}}_{s+1}^{-1}u_{s+1}}{t_{s}\tau_{s+1}} + \frac{\tilde{\mathbf{H}_{0}^{-1}}u_{0} \mathbf{1}_{w = 0}}{t_{s}\alpha_{1}}\notag \\
    & + \frac{1}{t_{s}}\sum_{k=1}^{s}\log(k+1)^w  \tilde{\mathbf{H}}_{k}u_{k} \left( \frac{1}{\tau_{k+1}} - \frac{1}{\tau_{k}} \right)  
\end{align}
Since $\tilde{\mathbf{H}}_{s}^{-1}$ converges almost surely to $I_{F}^{-1}$ which is positive, and $\left| \tau_{k+1}^{-1} - \tau_{k}^{-1} \right| \leq \alpha c_{\alpha}^{-1}k^{1-\alpha}$ and with the help of   Theorem \ref{theo::rate::lambda} and equation \eqref{equiv::ts}, one can check that 
%\MNT{Sorry, but this is not straightforward to me to see this. Any elaboration would be great} \textcolor{purple}{Is it more clear now? If not, I can detail it more. MNT: Much clearer now, thanks Antoine}
  \[
  \left\| \frac{1}{t_{s}} \sum_{k=0}^{s}\log(k+1)^w   \frac{\tilde{\mathbf{H}}_{k}^{-1} u_{k} - \tilde{\mathbf{H}}_{k+1}^{-1}u_{k+1}   }{\tau_{k+1}} \right\|^{2} = \mathcal{O} \left( \frac{\log s}{s^{2-\alpha}} \right) \quad a.s.
  \]
  which is negligible since $\alpha < 1$. For any $\delta > 0$, consider the event 
  \[E_{k} = \left\lbrace \left\| \lambda^{(k)} - \lambda^{*} \right\|^{2} \leq \frac{\log k^{1+\delta}}{k^\alpha}, \left\| \overline{\lambda}^{(k)} - \lambda^{*} \right\|^{2} \leq \frac{\log k^{1+\delta}}{k^\alpha} \right\rbrace.\]
  Thanks to Theorem \ref{theo::rate::lambda}, $\mathbf{1}_{E_{k}^{C}}$ converges almost surely to $0$, and consequently
  \[
  \left\| \frac{1}{t_{s}}\sum_{k=0}^{s} \log(k+1)^w \left( \tilde{\mathbf{H}}_{k+1}^{-1} - \tilde{\mathbf{H}}_{k}^{-1} \right) \frac{u_{k+1}}{\gamma_{k+1}} \mathbf{1}_{E_{k+1}^{C}} \right\|^{2} = \mathcal{O} \left( \frac{1}{s^2 \log s^{2w}} \right) \quad a.s.
  \]
  In addition, 
\begin{align*}
    \left\| \tilde{\mathbf{H}}_{k+1}^{-1} - \tilde{\mathbf{H}}_{k}^{-1} \right\|_{op} &  \leq \left\| \tilde{\mathbf{H}}_{k+1}^{-1} \right\|_{op} \left\| \tilde{\mathbf{H}}_{k}^{-1} \right\|_{op} \left\| \tilde{\mathbf{H}}_{k+1}  - \tilde{\mathbf{H}}_{k} \right\|_{op} \\
    & \leq \left\| \tilde{\mathbf{H}}_{k+1}^{-1} \right\|_{op} \left\| \tilde{\mathbf{H}}_{k}^{-1} \right\|_{op} \Big( \frac{1}{k+1} \left\| \tilde{\mathbf{H}}_{k} \right\|_{op} +  \underbrace{\frac{1}{k+1} \left\| \nabla_{\lambda} \log q_{\overline{\lambda}^{(k)}} \left( \overline{\theta}_{ k+1} \right) \right\|^{2} + c_\beta(k+1)^{-\beta}\left\| Z_{k+1}  \right\|^{2}}_{=\tilde{r}_{k+1}}  \Big)
\end{align*}
Then, since $\tilde{\mathbf{H}}_{k}^{-1}$ converges almost surely to a positive matrix, one can easily check that
\[
\left\| \frac{1}{t_{s}}\sum_{k=0}^{s} \log(k+1)^w \frac{1}{k+1} \left\| \tilde{\mathbf{H}}_{k+1}^{-1} \right\|_{op} \left\| \tilde{\mathbf{H}}_{k}^{-1} \right\|_{op}  \left\| \tilde{\mathbf{H}}_{k+1} \right\|_{op} \frac{ \left\| u_{k} \right\|}{\tau_{k+1}} \mathbf{1}_{E_{k+1}}\mathbf{1}_{E_{k}} \right\|^{2} = \mathcal{O} \left( \frac{\log s^{1  + \delta }}{s^{2- \alpha }} \right) \quad a.s.
\]
In addition, considering the filtration  $ \tilde{\mathcal{F}}_{k} $,
and by hypothesis, one has
\begin{align}
\notag  &   \frac{1}{t_{s}}\sum_{k=0}^{s}  \frac{\log (k+1)^{w}}{k+1} \frac{\left\| u_{k+1} \right\|}{\tau_{k+1}} \left\| \tilde{\mathbf{H}}_{k+1}^{-1} \right\|_{op} \left\| \tilde{\mathbf{H}}_{k}^{-1} \right\|_{op}     \mathbf{1}_{E_{k+1}}  \tilde{r}_{k+1}  \\
\notag   & \leq  \underbrace{\frac{1}{t_{s}}\sum_{k=0}^{s} \frac{\log (k+1)^{w}}{k+1} \frac{\left\| u_{k+1} \right\|}{\tau_{k+1}} \left\| \tilde{\mathbf{H}}_{k+1}^{-1} \right\|_{op} \left\| \tilde{\mathbf{H}}_{k}^{-1} \right\|_{op}     \mathbf{1}_{E_{k+1}} \left( \frac{{C_{0}'}^{\frac12} + {C_{1}'}^{\frac12}(\mathcal{L}(\overline{\lambda}^{(k+1)}) - \mathcal{L}\left( \lambda^{*} \right))}{k+1} + dc_\beta(k+1)^{-\beta} \right)}_{=:(*)}  \\
\notag& +   \underbrace{\frac{1}{t_{s}}\sum_{k=0}^{s} \frac{\log (k+1)^{w}}{k+1} \frac{\left\| u_{k+1} \right\|}{\tau_{k+1}} \left\| \tilde{\mathbf{H}}_{k+1}^{-1} \right\|_{op} \left\| \tilde{\mathbf{H}}_{k}^{-1} \right\|_{op}     \mathbf{1}_{E_{k+1}} \tilde{\xi}_{k+1}}_{=: (**)}
\end{align}
where $\tilde{\xi}_{k+1} := \tilde{r}_{k+1} - \mathbb{E} \left[\tilde{r}_{k+1} | \tilde{\mathcal{F}}_{k }\right]$  is a martingale difference. Then, since $\overline{\lambda}^{(k)}$ converges almost surely to $\lambda^{*}$, it comes (at least)
\[
\| (*) \|^{2} = \mathcal{O} \left( \frac{\log s^{1  + \delta }}{s^{2- \alpha }} \right) \quad a.s.
\]
In a same way, since $\tilde{\xi}_{k}$ is a martingale difference satisfying
\[
\mathbb{E} \left[ \left\| \tilde{\xi}_{k+1} \right\|^{2} | \tilde{\mathcal{F}}_{k }\right] \leq \frac{2}{(k+1)^{2}} \left( C_{0}' + C_{1}\left(  \mathcal{L}\left( \overline{\lambda}^{(k+1)} \right) - \mathcal{L}\left( \lambda^{*} \right) \right)^2  \right) + 6d^{2} c_\beta^2(k+1)^{-2\beta} 
\]
applying Theorem 6.2 in \cite{cenac2020efficient}, it comes that at least
\[
\| (**) \|^{2} = \mathcal{O} \left( \frac{\log s^{1  + \delta }}{s^{2- \alpha }} \right) \quad a.s.
\]
% Thanks to Theorem \ref{theo::rate::lambda} and equation \eqref{equiv::ts}, it comes
% \[
% \left\| \frac{u_{s+1} \log (s+1)^{w}}{\tau_{s+1}t_{s}}  \right\|^{2} = O \left(  \frac{\log s}{s^{2-\alpha}}\right) \quad a.s
% \]
% which is negligible as soon as $\alpha <1$. In a same way
% \[
% \left\| \frac{u_{0}\mathbf{1}_{w = 0}}{\alpha_{1}t_{s}} \right\|^{2} = O \left( \frac{1}{s^{2}} \right) \quad a.s
% \]
% In addition, remark that
% \[
% \left| \frac{\log(k+1)^w}{\tau_{k+1}} - \frac{\log (k)}{\tau_{k}} \right| \leq \max \lbrace \alpha , w \rbrace k^{\alpha -1} \log(k+1)^{w}
% \]
% so that
% \[
% \left\|  \frac{1}{t_{s}}\sum_{k=1}^{s} u_{k} \left( \frac{\log(k+1)^w}{\tau_{k+1}} -  \frac{\log(k )^w}{\alpha_{k }} \right) \right\|^{2} = O \left( \frac{\log s}{s^{2-\alpha}} \right) \quad a.s
% \]

\medskip
\noindent\textbf{Rate of convergence of $A_{2,s}$. } Thanks to inequality \eqref{eq::delta} and with the help of Theorem \ref{theo::rate::lambda}, it comes
\[
\left\| \delta_{k} \right\|  = \mathcal{O} \left( \frac{\log k }{k^{\alpha}} \right) \quad a.s.
\]
Then, one can check that
\[
\left\| A_{2,s} \right\|^{2} = \mathcal{O} \left( \frac{\log s^2}{s^{2\alpha}} \right) \quad a.s
\]
which is negligible as soon as $\alpha > 1/2$.

\medskip

\noindent\textbf{Rate of convergence of $M_{2,s}$. } Let us denote $t_{s}' = \sum_{k=0}^{s} \log(k+1)^{2w}$, then considering the filtration  $ {\mathcal{F}}_{k}$,
\begin{align*}
\frac{1}{t_{s}'} \sum_{k=0}^{s} \log(k+1)^{2w} \mathbb{E} \left[\xi_{k+1}\xi_{k+1}^{T} |\mathcal{F}_{k}\right]   = 
\frac{1}{t_{s}'} \sum_{k=0}^{s} \log(k+1)^{2w} \nabla_{\lambda} \mathcal{L}\left( \lambda^{(k)} \right) \nabla_{\lambda} \mathcal{L}\left( \lambda^{(k)} \right) ^{T}\\
+\underbrace{\frac{1}{t_{s}'} \sum_{k=0}^{s} \log(k+1)^{2w}  \mathbb{E} \left[ \frac{1}{B^2}\sum_{i=1}^{B} \ell\left( \theta_{k+1,i},\lambda^{(k)} \right)\sum_{i=1}^{B} \ell\left( \theta_{k+1,i},\lambda^{(k)} \right)^{T} | {\mathcal{F}}_{k} \right]}_{=: \langle M \rangle_{s}}. 
\end{align*}
Since the gradient of $\mathcal{L}$ is continuous at $\lambda^{*}$ and since $\lambda^{(k)}$ converges almost surely to $\lambda^{*}$, it comes that
\[
\frac{1}{t_{s}'} \sum_{k=0}^{s} \log(k+1)^{2w} \nabla_{\lambda} \mathcal{L}\left( \lambda^{(k)} \right) \nabla_{\lambda} \mathcal{L}\left( \lambda^{(k)} \right) ^{T}   \xrightarrow[s\to + \infty]{a.s.} 0 .
\]
In addition, since $\theta_{k+1,1},\ldots, \theta_{k+1,B}$ are i.i.d, it comes
\begin{align*}
\langle M \rangle_{s} &= \frac{1}{B t_{s}'} \sum_{k=0}^{s} \log(k+1)^{2w} \underbrace{  \mathbb{E}_{q_{\lambda^{(k)}}} \left[ \ell \left(\theta , \lambda^{(k)} \right)\ell \left(\theta , \lambda^{(k)} \right)^{T} \right]}_{=\Sigma\left( \lambda^{(k)} \right)} \\
&+ \underbrace{  \frac{B-1}{B t_{s}'} \sum_{k=0}^{s}  \log(k+1)^{2w} \nabla_{\lambda} \mathcal{L}\left( \lambda^{(k)} \right) \nabla_{\lambda} \mathcal{L}\left( \lambda^{(k)} \right) ^{T} }_{\xrightarrow[s\to + \infty]{as} 0} .
\end{align*}
and since $\lambda^{(k)}$ converges almost surely to $\lambda^{*}$ and since $\Sigma$ is continuous, it comes
\[
\langle M \rangle_{s} \xrightarrow[s\to + \infty]{a.s} \frac{1}{B}\Sigma \left( \lambda^{*} \right) .
\]
Then, thanks to the law of large numbers, one has
\[
\left\| \sum_{k=0}^{s} \log(k+1)^{w} \xi_{k+1} \right\|^{2} = \mathcal{O} \left( \sum_{k=0}^{s} \log(k+1)^{2w} \log \left( \sum_{k=0}^{s} \log(k+1)^{2w}  \right) \right) \quad a.s
\]
ans since
\[
\sum_{k=0}^{s} \log(k+1)^{2w}  \sim \frac{1}{s+1}\left( \sum_{k=0}^{s} \log(k+1)^{w} \right)^2 \sim (s+1)\log(s+1)^{2w}, 
\]
it comes that
\[
M_{2,s} = \mathcal{O} \left( \frac{\log s}{s} \right) \quad a.s.
\]
In addition, with the help of a Central Limit Theorem for Martingales \citep{duflo2013random}, it comes that
\[
\sqrt{Bs} M_{2,s} \xrightarrow[s\to + \infty]{\mathcal{L}}\mathcal{N}\left(0, L^{-1}\Sigma(\lambda^{*} ) L^{-1} \right) ,
\]
which concludes the proof.
\end{proof}

%\bibliographystyle{amsalpha}
%\bibliographystyle{mynatbib}
%\bibliographystyle{plainnat}
%\bibliography{LV}

\appendix
\section{Appendix}

\subsection{Proof of Theorem \ref{thm:H}}\label{Appendix:Prooof of Basic averaging theorem}
\begin{proof}
First, Riccati's equation (also known as the Sherman-Morrison formula) for matrix inversion in \cite[p. 96]{duflo2013random} states that for any $d\times d$ invertible matrix $S$, $d\times d$ invertible matrix $T$,  $p\times d$ matrix $U$, and $p\times d$ matrix $V$, one has that the matrix $S+UTV$ is invertible  if $VS^{-1}U+T^{-1}$ is invertible and that, in this case, 
	$$
(S+UTV)^{-1}=S^{-1}-S^{-1}U(VS^{-1}U+T^{-1})^{-1}VS^{-1}.
$$
We will prove the claim by induction.
Obviously it is true when $s=0$.
Now assume that it is true for $s>0$.
We have
\begin{align*}
 H_{s+1}^{-1}&= H^{-1}_s-\left(1+\phi^\top_{i+1} H_{s}^{-1}\phi_{s+1}\right)^{-1}
 H_s^{-1}\phi_{s+1}\phi_{s+1}^\top  H_{s}^{-1}\\
&=\left( H_s+\phi_{s+1}\phi_{s+1}^\top \right)^{-1}\\
&=\left( H_0+\sum_{j=1}^s \phi_j\phi_j^\top+\phi_{s+1}\phi_{s+1}^\top\right)^{-1}\\
&=\left(H_0+\sum_{j=1}^{s+1} \phi_j\phi_j^\top\right)^{-1}.
\end{align*}
As a result, $
H_{s+1}= H_0+\sum_{j=1}^{s+1} \phi_j\phi_j^\top.
$
From this, it is clear that $H_{k}^\top= H_k$ for $k\in\mathbb N$.
Lastly, we have for a vector $v$, $v^\top(\phi_j\phi_j^\top)v=|v^\top\phi_j|^2\geq 0$, the sum of positive definite matrices is positive definite, and the inverse of a positive definite matrix is positive definite. Hence the positivity and symmetry of $ H_s$ are
preserved.

\noindent Second, we have

\begin{align*}
\frac{1}{s}  H_s&=\frac{1}{s} H_0+\frac{1}{s}
\sum_{j=1}^s\nabla_\lambda\log q_\lambda(\theta_j) (\nabla_\lambda\log q_\lambda(\theta_j))^\top\\
&\longrightarrow I_F(\lambda), \quad\mbox{a.s}.
\end{align*}

\end{proof}

\subsection{Proof of \eqref{Riccati-twice-algorithm for A}}\label{Proof of Riccati-twice-algorithm}
Recall that, for a fixed $\lambda$, 
\[
A_{s} = A_{s}(\lambda) = H_0 + \sum_{j=1}^{s}\nabla_\lambda\log q_{\lambda }(\theta_{j}) (\nabla_\lambda\log q_{\lambda }(\theta_{j }))^\top + c_\beta\sum_{j=1}^{s} j^{-\beta}Z_{j}Z_{j}^\top
\]
where $Z_{1} , \ldots , Z_{s}$ are independent standard Gaussian vectors, $c_{\beta}\geq 0$ and $\beta \in (0, \alpha -1/2)$, $\alpha\in(1/2,1)$. 
%Let $\mathbf{A}_{s} = {A}_{s}/s$. 
One can update $A_{s+1}^{-1}$ using the following scheme:

\begin{thm}\label{thm:H1}
Let $\phi_s=\nabla_\lambda\log q_\lambda(\theta_s)$  and 
\begin{align*}
A_{s+ \frac{1}{2}}^{-1} & = A_{s}^{-1} - \left( 1+ \phi_{s+1}^{\top} A_{s}^{-1} \phi_{s+1} \right)^{-1}A_{s}^{-1} \phi_{s+1}\phi_{s+1}^{\top} A^{-1}_{s},
\end{align*}
then
\begin{align*}
A_{s+ 1}^{-1} & = A_{s+ \frac{1}{2}}^{-1} - c_{\beta}(s+1)^{-\beta} \left( 1+ c_{\beta}(s+1)^{-\beta} Z_{s+1}^{\top}A_{s+ \frac{1}{2}}^{-1}Z_{s+1} \right)^{-1} A_{s+ \frac{1}{2}}^{-1}Z_{s+1}Z_{s+1}^{\top}A_{s+ \frac{1}{2}}^{-1}.
\end{align*}
In particular, the positivity and symmetry of $A_{s}$'s is
preserved. Moreover,
\begin{align*}
\frac{1}{s} A_{s}\longrightarrow I_F(\lambda)\quad\mbox{a.s.}
\end{align*}
\end{thm}

\begin{proof}
We will prove the claim by induction as in the proof of Theorem \ref{thm:H}.
The claim is obviously true when $s=0$. Assume that it is true for some $s>0$.
By Riccati's formula,
\begin{align}\label{Riccati-twice-Hs051}
A_{s+ \frac{1}{2}}^{-1} & = A_{s}^{-1} - \left( 1+ \phi_{s+1}^{\top} A_{s}^{-1} \phi_{s+1} \right)^{-1}A_{s}^{-1} \phi_{s+1}\phi_{s+1}^{\top} A^{-1}_{s}
=\left( A_{s} +\phi_{s+1}\phi_{s+1}^\top\right)^{-1}.
\end{align}
Similarly,
\begin{align}\label{Riccati-twice-Hs052}
A_{s+ 1}^{-1} & = A_{s+ \frac{1}{2}}^{-1} - c_{\beta}(s+1)^{-\beta} \left( 1+ c_{\beta}(s+1)^{-\beta} Z_{s+1}^{\top}A_{s+ \frac{1}{2}}^{-1}Z_{s+1} \right)^{-1} A_{s+ \frac{1}{2}}^{-1}Z_{s+1}Z_{s+1}^{\top}A_{s+ \frac{1}{2}}^{-1}\notag\\
&=\left( A_{s+ \frac{1}{2}}+c_{\beta}(s+1)^{-\beta} Z_{s+1}Z_{s+1}^\top\right)^{-1}
\end{align}
Plugging \eqref{Riccati-twice-Hs051} into \eqref{Riccati-twice-Hs052}, we have
\begin{align*}
A_{s+ 1}^{-1}&=\left(A_{s} +\phi_{s+1}\phi_{s+1}^\top+ c_{\beta}(s+1)^{-\beta} Z_{s+1}Z_{s+1}^\top\right)^{-1}\\
%&=\left(A_{0}+\sum_{k=1}^{s+1%}\phi_k\phi_k^\top+\sum_{k=1}^{s+1}\frac{c_{\beta}}%{k^\beta} Z_{k}Z_{k}^\top \right)^{-1}\\
&=\left(H_0+\sum_{k=1}^{s+1}\phi_k\phi_k^\top+c_{\beta}\sum_{k=1}^{s+1}k^{-\beta} Z_{k}Z_{k}^\top \right)^{-1}.
\end{align*} 
Hence,
\[A_{s+ 1}=H_0+\sum_{k=1}^{s+1}\phi_k\phi_k^\top+c_{\beta}\sum_{k=1}^{s+1}k^{-\beta} Z_{k}Z_{k}^\top.\]
From this, it is clear that $A_{s}^\top=A_{s}$ and $A_{s}>0$ for $s\in\mathbb N$.
Hence the positivity and symmetry of $A_{s}$ are preserved. 

Second, we have

\begin{align*}
\frac{1}{s}A_{s}=\frac{1}{s}H_0+\frac{1}{s}\sum_{k=1}^{s}\phi_k\phi_k^\top+\frac{c_{\beta}}{s}\sum_{k=1}^{s}k^{-\beta} Z_{k}Z_{k}^\top.
\end{align*}
It can be seen that
\begin{align*}
\frac{1}{s}H_0+\frac{1}{s}\sum_{k=1}^{s}\phi_k\phi_k^\top\to 0+\mathbb E[\phi\phi^\top]=
\mathbb{E}_{q_\lambda}\left(\nabla_\lambda\log q_\lambda(\theta)\nabla_\lambda\log q_\lambda(\theta)^\top \right)
=I_F(\lambda)\quad\mbox{a.s.}
\end{align*}
Consider the term $\frac{c_{\beta}}{s}\sum_{k=1}^{s}k^{-\beta} Z_{k}Z_{k}^\top$, we have
\begin{align*}
\frac{c_\beta}{s}\sum_{k=1}^{s}k^{-\beta} Z_{k}Z_{k}^\top=\frac{c_\beta}{s}\left(\sum_{k=1}^s k^{-\beta}\right)\frac{1}{\sum_{k=1}^s k^{-\beta}}\sum_{k=1}^{s}\frac{1}{k^\beta} Z_{k}Z_{k}^\top.
\end{align*}
Similar to Lemma 6.1 in \cite{bercu2020efficient} and recall that $Z_k$'s are independent copies of $Z\sim\mathcal N(0,\mathbb I_D)$,  we have 
\begin{align*}
\frac{1}{\sum_{k=1}^s k^{-\beta}}\sum_{k=1}^{s}k^{-\beta} Z_{k}Z_{k}^\top \to \mathbb E[ZZ^\top] =\mathbb I_D,\quad \mbox{a.s.}
\quad \mbox{as}\quad s\to\infty.
\end{align*}
Also note that
\begin{align*}
\frac{1}{s}\sum_{k=1}^s k^{-\beta}=\frac{1}{s^\beta}\frac{1}{s^{1-\beta}}\sum_{k=1}^s k^{-\beta}\to 0\quad\mbox{as}\quad s\to\infty,
\end{align*}
due to the fact that $\frac{1}{s^{1-\beta}}\sum_{k=1}^s k^{-\beta}\to \frac{1}{1-\beta}$.
As a result, we have 
\begin{align*}
\frac{1}{s}\sum_{k=1}^{s}k^{-\beta} Z_{k}Z_{k}^\top\to 0\quad\mbox{a.s.} \quad \mbox{as}\quad s\to\infty.
\end{align*}
Therefore 
\begin{align*}
\frac{1}{s} A_{s}\longrightarrow I_F(\lambda) \quad\mbox{a.s} \quad \mbox{as}\quad s\to\infty.
\end{align*}
This completes the proof.
\end{proof}

\subsection{Calculation of the Hessian}\label{app: appendix on Hessian}
\noindent Using the log derivative trick,
$\nabla q_\lambda(\theta) =q_\lambda(\theta)\nabla \log q_\lambda(\theta) $,
it can be seen that
\begin{align*}
\nabla^2_\lambda q_\lambda(\theta)=\nabla^2_{\lambda}\log q_{\lambda}(\theta) q_\lambda(\theta)+ (\nabla_{\lambda}\log q_{\lambda}(\theta))^2 q_\lambda(\theta).
\end{align*}
Note that $h_{\lambda}(\theta)$ is a scalar. Hence
\begin{align*}
\nabla_\lambda^2\LB(\lambda)&=\nabla_\lambda\Big(\mathbb E_{q_\lambda}\left[\nabla_{\lambda}\log q_{\lambda}(\theta) h_{\lambda}(\theta)\right] \Big)\\
&=\nabla_\lambda\int \nabla_{\lambda}\log q_{\lambda}(\theta) h_{\lambda}(\theta) q_\lambda(\theta)d\theta\\
&=\int \nabla^2_{\lambda}\log q_{\lambda}(\theta) h_{\lambda}(\theta) q_\lambda(\theta)d\theta+\int  \nabla_\lambda h_{\lambda}(\theta) \nabla_{\lambda}\log q_{\lambda}(\theta)^\top q_\lambda(\theta)d\theta\\
&\phantom{cccc}+\int \nabla_{\lambda}\log q_{\lambda}(\theta) h_{\lambda}(\theta) \nabla_\lambda q_\lambda(\theta)d\theta\\
&=-\int  \nabla_\lambda \log q_{\lambda}(\theta) \nabla_{\lambda}\log q_{\lambda}(\theta)^\top q_\lambda(\theta)d\theta\\
&\phantom{cccc}+\int \left(\nabla^2_{\lambda}\log q_{\lambda}(\theta) q_\lambda(\theta)+ (\nabla_{\lambda}\log q_{\lambda}(\theta))^2 q_\lambda(\theta) \right)h_\lambda(\theta) d\theta\\
&= -\mathbb{E}[\nabla_\lambda q_\lambda(\theta) \nabla_\lambda q_\lambda(\theta)^\top]
-\int\nabla_\lambda^2 q_\lambda(\theta)\cdot\left(\log q_\lambda(\theta)-\log p(y,\theta) \right)d\theta\\
&=-I_F(\lambda)-\int\nabla_\lambda^2 q_\lambda(\theta)\cdot\left(\log q_\lambda(\theta)-\log p(y,\theta) \right)d\theta\\
&=-I_F(\lambda)-\int\nabla_\lambda^2 q_\lambda(\theta)\cdot\left(\log q_\lambda(\theta)-\log p(\theta|y) \right)d\theta,
\end{align*}
where we have used the fact that $\nabla q_\lambda(\theta) =q_\lambda(\theta)\nabla \log q_\lambda(\theta) $
in the third line. Recall that $\mathcal L(\lambda)=-\LB(\lambda)$, therefore
\[\nabla_\lambda^2\mathcal{L}(\lambda)=I_F(\lambda)+\int\nabla_\lambda^2 q_\lambda(\theta)\cdot\big(\log q_\lambda(\theta)-\log p(\theta|y) \big)d\theta.\]

\subsection{Convergence rate to Hessian}\label{Appendix:Convergence rate to Hessian}
In this section we consider the convergence rate of 
$\tilde{\mathbf{H}}_{s}$ to $I_F:=I_F(\lambda^*)$ and $\nabla_\lambda^2\mathcal L(\lambda^*)$.
\begin{thm}\label{FisherInformationMatrixconvergence}
Suppose that $\lambda^{(k)}$ converges almost surely to $\lambda^{*}$, and $\frac{\norm{\nabla_\lambda^2 q_{\lambda^*}(\theta)}}{q_{\lambda^*}(\theta)}
$ is bounded above by $M>0$. Then
$$\norm{\tilde{\mathbf{H}}_{s} -I_F }=\mathcal O \left( \max \left\lbrace c_{\beta}s^{-\beta}, s^{-1} \right\rbrace \right) \quad a.s.$$ 
Furthermore, under the assumptions in Theorem 2.1 of \cite{zhang2020convergence}, the following holds true 
$$\norm{\tilde{\mathbf{H}}_{s} - \nabla_\lambda^2\mathcal L(\lambda^*)}=\mathcal O\left( \max \left\lbrace c_{\beta}s^{-\beta}, s^{-1} ,n^{-1}\right\rbrace \right) \quad a.s.$$ 
with $n$ the size of the data.
\end{thm}

%\textcolor{purple}{Is it usual that $I_F(\lambda^{*}) = \nabla^{2} \mathcal{L}\left( \lambda^{*} \right)$? If not, I think that we just need to suppose that the %Fisher Information is positive to get the rates of convergence. }\blue{This is not often the case. Normally we restrict our search on a computable spaces %$\mathcal Q$ such as the normal distribution class}

\begin{proof}
First recall that
\begin{align*}
\nabla_\lambda^2\mathcal L(\lambda)=I_F(\lambda)+\int\nabla_\lambda^2 q_\lambda(\theta)\cdot\left(\log q_\lambda(\theta)-\log p(\theta|y) \right)d\theta.
\end{align*}
and that 
\[
\tilde{\mathbf{H}}_{s} = \frac{1}{s} \left(  H_0 + \sum_{k=0}^{s-1}\nabla_\lambda\log q_{\bar\lambda^{(k )}}(\theta_{k+1}) (\nabla_\lambda\log q_{\bar\lambda^{(k)}}(\theta_{k+1}))^\top + c_\beta\sum_{k=1}^{s} k^{-\beta}Z_{k}Z_{k}^{T} \right) 
\]
Hence 
\begin{align}\label{eq: Hessian difference term}
\norm{\tilde{\mathbf{H}}_{s} - \nabla_\lambda^2\mathcal L(\lambda^*)}&\leq \norm{ \frac{1}{s} \left(  H_0 + \sum_{k=0}^{s-1}\nabla_\lambda\log q_{\bar\lambda^{(k )}}(\theta_{k+1}) (\nabla_\lambda\log q_{\bar\lambda^{(k)}}(\theta_{k+1}))^\top + c_\beta\sum_{k=1}^{s} k^{-\beta}Z_{k}Z_{k}^{T} \right) -I_F}\notag\\
&+\norm{\int\nabla_\lambda^2 q_{\lambda^*}(\theta)\cdot\left(\log q_{\lambda^*}(\theta)-\log p(\theta|y) \right)d\theta}
\end{align}
First, similar to the proof of Theorems \ref{theo::wasn} and \ref{thm:H1}, we have
\begin{align*}
&\norm{ \frac{1}{s} \left(  H_0 + \sum_{k=0}^{s-1}\nabla_\lambda\log q_{\bar\lambda^{(k )}}(\theta_{k+1}) (\nabla_\lambda\log q_{\bar\lambda^{(k)}}(\theta_{k+1}))^\top + c_\beta\sum_{k=1}^{s} k^{-\beta}Z_{k}Z_{k}^{T} \right) -I_F}\\
&\leq\norm{\frac{1}{s} \sum_{k=0}^{s-1}\nabla_\lambda\log q_{\bar\lambda^{(k )}}(\theta_{k+1}) (\nabla_\lambda\log q_{\bar\lambda^{(k)}}(\theta_{k+1}))^\top-\mathbb E[\nabla_\lambda\log q_\lambda(\theta) (\nabla_\lambda\log q_\lambda(\theta))^\top]}\\
&+\norm{ \frac{1}{s} \left(  H_0 + c_\beta\sum_{k=1}^{s} k^{-\beta}Z_{k}Z_{k}^{T} \right)}\\
&= O \left( \max \left\lbrace c_{\beta}s^{-\beta}, s^{-1} \right\rbrace \right) \quad a.s.
\end{align*}
\noindent For the second term of\eqref{eq: Hessian difference term},
\begin{align*}
& \norm{\int\nabla_\lambda^2 q_{\lambda^*}(\theta)\cdot\left(\log q_{\lambda^*}(\theta)-\log p(\theta|y) \right)d\theta}\\
&\leq\int\frac{\norm{\nabla_\lambda^2 q_{\lambda^*}(\theta)}}{ q_{\lambda^*}(\theta)} q_{\lambda^*}(\theta)|\left(\log q_{\lambda^*}(\theta)-\log p(\theta|y) \right)|d\theta\\
&\leq M \int q_{\lambda*}(\theta)\big|\log q_{\lambda*}(\theta)-\log p(\theta|y) \big|d\theta\\
&=M \KL(q_{\lambda^*}\|p(\cdot|y))\\
&\leq MC/n
\end{align*}
where the last inequality is from Theorem 2.1 of \cite{zhang2020convergence}.
By combining the estimates together, the desired result
is followed immediately.
\end{proof}

\subsection{Technical details for Example 4}\label{Appendix:Example 5}
Denote by $p_\epsilon(\epsilon)$, $p_Z(z)$ and $q_\lambda(\theta)$ the density function of random vectors $\epsilon$, $Z$ and $\theta$, respectively, $\lambda=(\text{vec}(W_1)^\top,b_1^\top,\text{vec}(W_2)^\top,b_2^\top)^\top$. Write $D$ %\red{you meant $d$?}\MNT{No, $d$ is the dimension of model parameter $\theta$, $D$ is the dimension of variational parameter $\lambda$}
for the length of $\lambda$.
We have that 
\begin{align*}
p_Z(z)&=p_\epsilon(\epsilon)|\frac{\partial\epsilon}{\partial z}|=p_\epsilon(\epsilon)\prod_{i=1}^d\frac{1}{h'(h^{-1}(z_i))},\;\;\epsilon=W_1^\top(h^{-1}(z)-b_1)\\
&=p_\epsilon\big(W_1^\top(h^{-1}(z)-b_1)\big)\prod_{i=1}^d\frac{1}{h'(h^{-1}(z_i))},
\end{align*}
and
\[q_\lambda(\theta)=p_Z\big(z=W_2^\top(\theta-b_2)\big).\]
If we use $q_\lambda(\theta)$ to approximate a posterior distribution with prior $p(\theta)$ and log-likelihood $\ell(\theta)$, the lower bound is
\begin{align*}
\LB(\lambda)%&=\E_{p_\epsilon}\left[\log\frac{p(\theta)L(\theta)}{q_\lambda(\theta)}\right]\\
&=\E_{p_\epsilon}\left[\log p(\theta)+\ell(\theta)-\log p_\epsilon(\epsilon)+\sum_1^d\log h'(h^{-1}(z_i))\right]
\end{align*}
where $z=h(W_1\epsilon+b_1),\;\;\theta=W_2z+b_2$. It's straighforward to estimate $\nabla_\lambda \LB(\lambda)$, if both $\log p(\theta)$ and $\ell(\theta)$ are differentiable in $\theta$. After some algebra
\[\frac{\partial\theta}{\partial\text{vec}(W_1)}=\epsilon^\top\otimes\Big(W_2\diag\big(h'(W_1\epsilon+b_1)\big)\Big),\;\;\frac{\partial\theta}{\partial b_1}=W_2\diag\big(h'(W_1\epsilon+b_1)\big) \]
\[\frac{\partial\theta}{\partial\text{vec}(W_2)}=z^\top\otimes I_d,\;\;\frac{\partial\theta}{\partial b_2}=I_d,\]
with $\otimes$ the Kronecker product. Also,
\[\frac{\partial z}{\partial\text{vec}(W_1)}=\epsilon^\top\otimes \diag\big(h'(W_1\epsilon+b_1)\big),\;\;\frac{\partial z}{\partial b_1}=I_d, \frac{\partial z}{\partial\text{vec}(W_2)}=0,\;\;\frac{\partial z}{\partial b_2}=0.\]
Then $\frac{\partial \theta}{\partial \lambda}$ and $\frac{\partial z}{\partial \lambda}$ are the $d\times D$ matrices formed by
\[\frac{\partial \theta}{\partial \lambda}=\left[\frac{\partial\theta}{\partial\text{vec}(W_1)},\frac{\partial\theta}{\partial b_1},\frac{\partial\theta}{\partial\text{vec}(W_2)},\frac{\partial\theta}{\partial b_2},\right],\;\;
\frac{\partial z}{\partial \lambda}=\left[\frac{\partial z}{\partial\text{vec}(W_1)},\frac{\partial z}{\partial b_1},\frac{\partial z}{\partial\text{vec}(W_2)},\frac{\partial z}{\partial b_2}\right].\]
It's now readily to compute the gradient of the lower bound
\begin{equation}\label{eq:repram trick VB}
\nabla_\lambda \LB(\lambda)=\E_{p_\epsilon}\left[\Big(\frac{\partial \theta}{\partial \lambda}\Big)^\top\nabla_\theta\Big(\log p(\theta)+\ell(\theta)\Big)+\Big(\frac{\partial z}{\partial \lambda}\Big)^\top h''\big(h^{-1}(z)\big)\right],
\end{equation}
which can be estimated by sampling from $p_\epsilon$.

In order to use the IFVB algorithm, we now derive the gradient $\nabla_\lambda \log q_\lambda(\theta)$. First, note that
\begin{equation*}
\log q_\lambda(\theta)=-\frac{d}{2}\log(2\pi)-\frac12(h^{-1}(z)-b_1)^\top W_1W_1^\top(h^{-1}(z)-b_1)-\sum_{i=1}^d\log h'(h^{-1}(z_i)),
\end{equation*}
where $z=W_2^\top(\theta-b_2)$. It is easy to see that
\begin{align*}
\frac{\partial\log q_\lambda(\theta)}{\partial W_1}&=-(h^{-1}(z)-b_1)(h^{-1}(z)-b_1)^\top W_1=-\epsilon\epsilon^\top W_1,\\
\frac{\partial\log q_\lambda(\theta)}{\partial b_1}&=W_1W_1^\top(h^{-1}(z)-b_1)=W_1\epsilon,\\
\frac{\partial\log q_\lambda(\theta)}{\partial W_2}&=(\theta-b_2)\delta_z^\top,\\
\frac{\partial\log q_\lambda(\theta)}{\partial b_2}&=-W_2\delta_z,
\end{align*}
where
\begin{align*}
\delta_z&=-\diag\Big(1/h'\big(h^{-1}(z)\big)\Big)W_1W_1^\top\big(h^{-1}(z)-b_1\big)-h''\big(h^{-1}(z)\big)\\
&=-\diag\Big(1/h'\big(h^{-1}(z)\big)\Big)W_1\epsilon-h''\big(h^{-1}(z)\big).    
\end{align*}
Vectorizing these four terms and stacking them together gives $\nabla_\lambda \log q_\lambda(\theta)$.

%\section{Inversion free on Manifold}
%\red{Only If we can do it}
%We will see in the examples from Section \ref{Numerical Example}
%that both exact natural gradient ascent
%and inversion free gradient ascent
%do not preserve the definite positive
%of the covariance matrix in Gaussian 
%approximation.
%This suggests us to
%extend the methodology
%proposed in this paper
%to a class of manifold. 
%\blue{The main references are \cite{tan2021analytic} (which
%is done for Euclidean spaces)
%and \cite{magris2022exact}}

%\bigskip
%\begin{center}
%{\large\bf SUPPLEMENTARY MATERIAL}
%\end{center}

%\begin{description}

%\item[Title:] Brief description. (file type)

%\item[R-package for  MYNEW routine:] R-package containing code to perform the diagnostic methods described in the article. The package also contains all datasets used as examples in the article. (GNU zipped tar file)

%\item[HIV data set:] Data set used in the illustration of MYNEW method in Section~ 3.2. (.txt file)

%\end{description}

%\section{BibTeX}

%We hope you've chosen to use BibTeX!\ If you have, please feel free to use the package natbib with any bibliography style you're comfortable with. The .bst file agsm has been included here for your convenience. 

\end{document}